\documentclass[11pt,  oneside, a4paper]{article}
\usepackage[utf8]{inputenc}
\usepackage[a4paper, total={6.5in,9in}]{geometry}

\usepackage{graphicx} 
\usepackage{bbm}
\usepackage{amsfonts}
\usepackage{enumitem}
\usepackage{fullpage}
\usepackage{amsmath}
\usepackage{hyperref}
\usepackage{mathtools}
\usepackage{lipsum}
\usepackage{blindtext}
\usepackage{titlesec}
\usepackage{amssymb}
\usepackage{amsthm}
\usepackage{rotating}
\usepackage{subcaption}
\usepackage{caption}
\usepackage{svg}
\usepackage{xcolor}
\usepackage{verbatim}
\usepackage{natbib}

\newtheorem*{theorem-nono}{Theorem}
\newtheorem{claim}{Claim}
\newtheorem{remark}{Remark}
\newtheorem{lemma}{Lemma}
\newtheorem{definition}{Definition}
\newtheorem{thm}{Theorem}

\newtheorem{example}{Example}
\newenvironment{proofsketch}{%
  
  \begin{proof}
}{
  \end{proof}
}

\newcommand{\juba}[1]{}
\newcommand{\kate}[1]{}
\newcommand{\kd}[1]{}
\newcommand{\raef}[1]{}
\newcommand{\ds}[1]{}
\newcommand{\annuo}[1]{}
\newcommand{\writenote}[1]{}

\renewcommand{\juba}[1]{\textcolor{blue}{[Juba: #1]}}
\renewcommand{\kate}[1]{\textcolor{purple}{[Kate: #1]}}
\renewcommand{\kd}[1]{\textcolor{purple}{[Kate: #1]}}
\renewcommand{\raef}[1]{\textcolor{brown}{[Raef: #1]}}
\renewcommand{\ds}[1]{\textcolor{red}{[Sen: #1]}}
\renewcommand{\annuo}[1]{\textcolor{teal}{[Annuo: #1]}}
\renewcommand{\writenote}[1]{\textcolor{magenta}{[Writing note: #1]}}

\newcommand{\footremember}[2]{%
    \footnote{#2}
    \newcounter{#1}
    \setcounter{#1}{\value{footnote}}%
}

\title{Data Sharing with Endogenous Choices\\over Differential Privacy Levels}
\author{Raef Bassily\footremember{trailer}{The Ohio State University. Email: bassily.1@osu.edu. Part of this work was done during his visit at Google.} 
\and Kate Donahue\footremember{something}{MIT, University of Illinois Urbana-Champaign. Email: kpd@illinois.edu}
\and Diptangshu Sen\footremember{alley}{Georgia Institute of Technology. Email: dsen30@gatech.edu. First author and corresponding author.}
\and Annuo Zhao\footremember{somethingelse}{Georgia Institute of Technology. Email: ann\_zhao@gatech.edu}
\and Juba Ziani\footremember{somethingelse2}{Georgia Institute of Technology. Email: jziani3@gatech.edu}
}

\date{\today}

\begin{document}

\maketitle

\begin{abstract}
Motivated by the rapid push to decentralize sharing of data, we study whether large-scale data sharing coalitions can form in a decentralized manner under differential privacy when players have heterogeneous privacy preferences. We first consider a fully decentralized data-sharing mechanism in which each player decides whether to participate and how much privacy noise to add locally to their sensitive data before sharing. Privacy choices induce a fundamental trade-off: higher privacy lowers individual privacy costs but reduces data utility and statistical accuracy for the coalition. These choices generate externalities across players, making both participation and privacy levels strategic. Our goal is to understand which coalitions are stable, how privacy choices shape equilibrium outcomes, and how fully decentralized data-sharing compares to a centralized, socially optimal benchmark when the number of players is large. We provide a comprehensive analysis across multiple privacy-cost regimes corresponding to different attack/observation models in differential privacy, showing that full decentralization is highly inefficient in terms of both social welfare and estimator accuracy. Surprisingly, we find that a simple partially decentralized mechanism (where players still retain participation agency, but a central designer chooses a fixed privacy noise level for everyone) closes this efficiency gap down to constant factors across all privacy-cost regimes.
\end{abstract}

\section{Introduction}\label{sec:intro}

Many organizations need access to large amounts of data in order to produce accurate statistical estimates or to build their own machine learning models. However, individual organizations often do not possess sufficient data on their own. This is an especially important bottleneck in settings with rare events or small populations, such as hospitals or research organizations studying rare diseases or healthcare interventions. In such environments, meaningful statistical inference is only possible if multiple data holders share data and pool their information~\citep{chen2025data}. As a result, data-sharing cooperatives or coalitions have become increasingly common and allow participants to aggregate data to build more powerful models and computations. For example, in healthcare, the European Network of Genetic and Genomic Epidemiology~\citep{engage} is a consortium of 24 leading research organizations that pooled genomic data of over 600,000 individuals across Europe, leading to better identification of genetic biomarkers that contribute to diseases like type 2 diabetes. In education, the Higher Education Data Sharing Consortium~\citep{heds} facilitates data sharing across institutions to build models for understanding and predicting student success. In risk modeling in finance and credit, the main credit bureaus pool consumer data to estimate their models for risk of defaulting on loans and mortgages. Experian, for example, has ``more than 12,000 data contributors across the nation"~\citep{experian}. 

Despite their potential benefits, data cooperatives can face significant obstacles~\citep{oracle,data-coops}, with a major one being \emph{data privacy}. Contributing sensitive data can lead to privacy leakages that enable inference or re-identification attacks, resulting in potential societal~\citep{genomicprivacy,nyt_location}, legal~\citep{netflix_lawsuit,ftc_facebook}, and economic harms~\citep{ghosh_selling_privacy,gkatzelis2015pricing,acemoglu2022}. 
Conversely, Differential Privacy (DP) has emerged as the gold standard for guaranteeing privacy for statistical analyses and model training~\citep{dwork2014algorithmic}. By injecting carefully calibrated noise into data, algorithms, or query responses, DP provides strong worst-case guarantees. It limits the impact of any single participant’s data on the released model or statistic, providing an information-theoretic notion of privacy. Importantly, differential privacy can help to reduce the privacy cost of participation, making data sharing more feasible. At the same time, privacy imposes a trade-off: the added noise degrades the accuracy of statistical estimates and computed models. This fundamental tradeoff between privacy protection and data utility raises the following question: \textit{how should we add noise to computations to balance accuracy and privacy?}

One way that data cooperatives often form is through \emph{decentralized} data-sharing environments as proposed, for example, by \citep{hardjono2019data}. There are several real-life initiatives that promote such decentralization. In the case of healthcare, new initiatives are emerging that enable data to remain locally controlled while supporting collaborative model development across institutions. One example is \citet{moonshot} which leverages a distributed network of hospital-owned machines to build privacy-preserving, real-time AI without requiring centralization of patient data. Related efforts in healthcare data cooperatives further emphasize patient-centric control, where individuals can decide how their data is shared and used for research and clinical purposes \citep{hafen2014health}, highlighting how ``members determine which data they want to share for example with doctors or to contribute to research for the benefit of their health and that of society''. In the context of power systems and the U.S. energy infrastructure, similar efforts emphasize decentralized and federated control over data, with architectures that enable coordination without centralizing ownership. For example, the U.S. Department of Energy's Genesis initiative highlights a move toward ``\emph{federated data governance}'' and ``standing up a federated governance model that (...) includes fine-grained access control, licensing, auditability, and digital provenance". More broadly, data space architectures such as the International Data Spaces framework promote interoperable, decentralized data sharing across energy stakeholders while preserving sovereignty over local data \citep{idsa_energy}.

Despite the push to decentralize data sharing, one aspect that makes the picture more challenging is that different players may want different privacy guarantees. This is because privacy attitudes are often heterogeneous \citep{ghosh_selling_privacy,personalized_privacy}, with different individuals or organizations having different understanding of privacy risks, risk tolerance, ownership of sensitive data, etc. This has been observed in large-scale surveys, with individuals self-reporting various degrees of privacy attitudes and concerns, as well as in empirical behavioral studies on privacy preferences and disclosure behavior \citep{pew2015privacy,berendt2005privacy} (see surveys for a more complete description of privacy attitudes at \citep{barth2017privacy,acquisti2015privacy}). While privacy preferences can be complex, some research (e.g. \cite{nanayakkara2023epsilon,dekel2022privacy}) shows that individuals can still understand and quantify their own (differential) privacy preferences and can trade off privacy with benefits obtained from sharing their data, meaning that such preferences do need to be taken into account.

These considerations lead to new fundamental questions: when data cooperatives form in a decentralized fashion and privacy preferences are endogenous, what does stable participation look like? Which coalitions can sustain themselves in equilibrium, and how does coalition composition depend on participants’ privacy costs? How do individual privacy decisions affect the accuracy of the shared estimate and overall social welfare? Crucially, how do outcomes under decentralized data sharing compare in terms of efficiency to an idealized world in which a centralized designer can choose participation and privacy levels to maximize social welfare? And, if they are indeed inefficient, can we do better while still giving participants some agency over decision-making?

In this paper, we consider a model of coalition formation for data-sharing under differential privacy that captures these considerations. Players have heterogeneous privacy costs and decide, autonomously, whether to participate in the coalition and the level of noise to add to their data. The shared statistic is then computed from the pooled, privatized data of the participants, and only those participants get a benefit from the final estimate. We study the equilibrium behavior induced by fully decentralized decisions, focusing on the structure of stable coalitions, the resulting privacy–accuracy tradeoffs, and the efficiency of decentralized outcomes in terms of social welfare and accuracy of the computed estimate. Our main results show that full decentralization is \textit{highly inefficient} compared to the baseline fully centralized setting over a large range of privacy cost regimes: this is primarily due to players choosing stringent privacy levels, which either prevents the formation of large coalitions or degrades their utility (in terms of social welfare and estimator variance) when they do come to exist. Also, we show that there exists a simple data sharing mechanism --- a \textit{partially decentralized} one --- where players still retain agency over participation decisions but the privacy level is chosen centrally, which can close the efficiency gap across all privacy regimes. Thus, \emph{the main source of inefficiency is letting players choose their own $\epsilon$}.  \\

\noindent 
\textbf{Summary of contributions.} We summarize our main contributions below: 
\begin{itemize}
\item In Section~\ref{sec:model}, we introduce a new model of decentralized data sharing through coalition formation under differential privacy. Our setting differs from standard coalition formation games because on top of choosing whether to participate in a coalition, our players must also decide on an endogenous privacy level. 
A central element to our model is that we introduce and reason across different regimes of privacy costs that correspond to different observation and attack models in differential privacy. One of our most important contributions is understanding how the specific privacy regime affects coalition formation under different levels of decentralization. To the best of our knowledge, we are the first to introduce such considerations regarding the privacy regime to the economics of data sharing literature. 
\item In Section~\ref{sec:baseline}, we study a baseline setting where a fully centralized designer controls both i) who participates in the coalition and ii) what privacy level is to be used by each participant, and minimizes social cost (equivalent to maximizing social welfare). We fully characterize i) the social cost and ii) the accuracy (measured by variance) of the final estimate computed by the \textit{socially optimal} coalition, across all privacy cost regimes, in particular, identifying regimes where it achieves non-trivial improvements across both metrics.  
\item In Section~\ref{sec:autonomous}, we provide a detailed analysis of the setting where players act in a fully decentralized fashion and form \textit{stable} coalitions. Here, players make simultaneous decisions across two dimensions: i) whether they want to participate in the coalition but also ii) what privacy level to add to their own data. Beyond characterizing the properties of the \textit{optimal stable} coalition and the best-case social cost and estimator accuracy achievable in this setting, we also provide insights on what stable coalitions look like (under different notions of stability), and when they are well-behaved given the parameters of the problem. 
\item In Section~\ref{sec:partial}, we consider a natural intermediate or hybrid mechanism of data sharing where the privacy level is imposed centrally (for example by an central designer or by a joint optimization by the coalition itself), but players still decide whether they want to participate in the coalition (hence, partially decentralized). Importantly, this allows us to separate the impact of players' i) autonomy of participation decisions and ii) autonomy in the choice of privacy level on social cost and accuracy of the final estimate. 
\item Finally, in Section~\ref{sec:pos}, we provide a detailed comparative analysis across all $3$ mechanisms. First, we characterize the ``Price of Stability'' between the fully decentralized mechanism and the centralized baseline mechanism in terms of social cost and variance; we highlight that i) the regimes of privacy costs, where the social cost and accuracy improvements (over no data sharing) are non-trivial, varies greatly between the two settings; and ii) even in cost regimes where both settings offer non-trivial improvements, a significant efficiency gap still persists. However, surprisingly, the partially decentralized mechanism is found to recover (up to constant factors) the optimal social cost and estimator accuracy achieved by the baseline mechanism across all privacy cost regimes. The central insight is that \emph{efficiency gaps arise not due to agency in participation decisions but because of how players selfishly pick their privacy levels}. Selfish players over-focus on their own privacy costs, not accounting for the \emph{joint impact} of their choice of privacy levels on the accuracy of the shared estimator.
\end{itemize}

\subsection{Related work}\label{sec:related}

\paragraph{Markets for Data.} There has been a significant amount of work studying and designing markets for data across computer science, operations research, and economics communities~\citep{anish_data_marketplace,liu2016learning,liu2017sequential,liu2018surrogate,dekel2010incentive,meir2011strategyproof,meir2012algorithms,perote2003impossibility,satml2023,fallah2023optimal,cummings2015accuracy,cai2015optimum,abernethy2015low,roth2012conducting,dataacquisition, gkatzelis2015pricing,bergemann2018design,bergemann2019information,shuran_yiling,choi2019privacy,acemoglu2022,bergemann2022economics,liao2022privacy,sen2025equilibria,fallah2024three,acemoglu2023good}. Here, there is a centralized designer who decides how to design the market and facilitates interactions between agents and data transactions. This line of work is vast, and studies data exchanges across different types of incentives around data sharing (receiving a payment vs receiving a useful model or service), different dis-utilities from sharing data (such as privacy concerns and increased competition or loss of economic power in game-theoretic settings), different levels of intervention from the centralized designer (designing all transactions vs only acting as a matching platform), to only name a few. For a comprehensive survey, refer to~\citet{bergemann_data_markets}.

\paragraph{Decentralized data coalition formation for data cooperatives.} Real-life data-sharing settings are often substantially more decentralized. In data cooperatives agents may have complete control over who they share data with, as well as the amount or quality of data to share, without the direct intervention of a centralized coordinator~\citep{hardjono2019data,hafen2014health,data-coops,ligett2023we}. There has also been recent work that explores the stability, economic viability and fairness aspects of data-sharing consortia~\citep{song2025existence,akrami2024theoretical,castro2023data}. To the best of our knowledge, none of these works account for what happens when privacy-conscious participants protect their data through sophisticated privacy mechanisms (like differential privacy) before sharing it across with the consortia.  

A related line of work has studied incentives in data-sharing in decentralized settings. From the game theoretic perspective, some works draw on tools from hedonic game theory, where the goal is to study which coalition structures are stable~\citep{bogomolnaia2002stability}. Works that specifically study decentralized data sharing from a theoretical perspective include~\citet{donahue2021model, hasan2021incentive, shadmy2024reimagining,chen2025r}. Other works also study incentives in coalition formation in settings beyond just data sharing (e.g., with other forms of cost sharing, such as compute, vaccines etc.)~\citep{Guazzone_2014,Anglano2018AGA,Nagalapatti_Narayanam_2021, blum2021one, donahue2021optimality,gui2026market}.

\paragraph{Data transactions under privacy constraints.} Our work also subscribes to the long literature on how to transact data under differential privacy constraints~\citep{buyingwithoutverification,ghosh_selling_privacy,fallah2022bridging,fallah2023optimal,satml2023}. 
One prominent direction in this literature involves designing mechanisms for truthful elicitation of private data from individuals or users with heterogenous privacy preferences.
In this problem setting, a key modeling choice concerns where privacy is enforced, either through noise addition on the data directly (local DP) or through addition of noise to the computation itself and its output (central DP). While some works like \citet{cummings2015accuracy} consider a local-DP model where the user adds noise to their data locally before sharing, others like \citet{satml2023,ghosh_selling_privacy,buyingwithoutverification} consider noise addition at the central level by the trusted curator. 
Works like~\citet{fallah2022bridging,fallah2023optimal} have tried to bridge these differences by considering mechanism design for both local and central-DP models. While in most of these cases, the decision of whether to participate and share data or not is fully exogenous and depends on each individual user, in some settings like that of~\citet{ghosh2013privacy}, it can be endogenous and depend the actions of the other participants. Despite some similarities of our work with all of this literature, there are some major differences: in particular, our primary setting is to study how players can spontaneously form coalitions for exchange of data to compute some population-level statistic, instead of sharing it on a third party or platform.

\paragraph{Non-cooperative game theory and coalition formation.}
The study of coalition formation lies at the intersection of non-cooperative and cooperative game theory. This line of work is based on the fundamental work on strategic behavior of \citet{Nash51}. This led to a rich literature on endogenous coalition formation in economics and operations research \citep{hart1983endogenous,rayvohra,nagarajan2009coalition,nguyen2015coalitional,gairing2019computing}. Classical cooperative notions of stability, such as the core, require that no subset of agents can jointly deviate to form a coalition that makes all its members strictly better off \citep{gillies1953some,aumann1961core,shapley1965balanced,debreuscarf}, and have been extensively studied in both abstract settings~\citep{pisinger1999core,deng1999algorithmic,bietenhader2006core} and applied settings~\citep{banerjee2001core,bichler2022core,goetzendorff2015compact}. More recent work has focused on stability notions grounded in unilateral deviations, particularly in hedonic games where agents have preferences over coalition membership. Prominent among these are Nash stability and individual stability, which capture whether agents can profitably deviate to another coalition with or without the consent of its members \citep{BOGOMOLNAIA2002201}. We specifically adopt the notion of Nash-stability (and its variant) in our work. Finally, to quantify the inefficiency induced by strategic behavior, we use the notion of the price of stability, which measures the best-case gap between equilibrium outcomes and the socially optimal solution~\citep{anshelevich2008price,feldman2012conflicting,christodoulou2015price,scarsini2018dynamic}. Another closely related and widely used concept is the price of anarchy which measures the worst-case gap~\citep{koutsoupias1999worst,roughgarden2002bad,roughgarden2005selfish}.

\paragraph{Novelty of our work.} Importantly, and to the best of our knowledge, we are the first to \textit{synthesize} several key research threads together: i) data exchanges, ii) differential privacy used as an economic knob to facilitate data sharing, and iii) coalition formation under different levels of decentralization. Additionally, as described in the contributions earlier, we are the first to study the impact of different privacy attack/observation models on the efficiency of data-sharing mechanisms.

\section{Differential Privacy Preliminaries}\label{sec:dp_prelims}

We now briefly review the elements of differential privacy essential for our analysis (for a comprehensive presentation, see \citet{dwork2014algorithmic}). The goal is to answer a real-valued query $q(.)$ on data $(x_1,\ldots,x_n)$ while maintaining the privacy of the data, where $x_i$ is the data point of agent $i$.

\paragraph{Neighboring datasets.} Differential Privacy (DP) asks: what happens in a counterfactual world in which a single agent changed their dataset entry? To reason about this, we need to define the concept of dataset and of neighboring datasets. A dataset is a vector $\mathbf{x}=(x_1,\dots,x_n)$, where $x_i$ denotes the data of agent $i$. Intuitively, neighboring datasets represent two possible worlds that differ only in the data contributed by a single agent; formally:

\begin{definition}[Neighboring Datasets]
Two datasets $\mathbf{x}$ and $\mathbf{x}'$ are said to be \emph{neighbors with respect to agent $i$} if they differ only in the $i$th coordinate, i.e., $x_j = x'_j$ for all $j \neq i$. 
\end{definition}

\paragraph{Differential privacy.} Differential privacy requires that the \emph{distribution} of outputs of a randomized algorithm be nearly indistinguishable on neighboring datasets. If outputs look similar across two neighboring datasets, then it is impossible to identify which dataset is the ``true'' one, hence to recover a specific agent's data point, only from a single observation\footnote{Privacy guarantees degrade smoothly in terms of the number of observations or queries thanks to powerful composition theorems.}. 

\begin{definition}[Differential Privacy]
A randomized mechanism $\mathcal{M}$ is \emph{$\epsilon$-differentially private with respect to agent $i$} if, for all pairs of neighboring datasets $\mathbf{x},\mathbf{x}'$ with respect to $i$, and for all measurable sets of outputs $\mathcal{O}$,
$
\Pr[\mathcal{M}(\mathbf{x}) \in \mathcal{O}]
\;\le\;
e^{\epsilon}\Pr[\mathcal{M}(\mathbf{x}') \in \mathcal{O}].
$
\end{definition}

The parameter $\epsilon > 0$ quantifies the strength of the privacy guarantee: smaller values of $\epsilon$ correspond to stronger privacy protection, as a smaller $\epsilon$ requires the distribution of outputs on neighboring datasets to be harder to distinguish. $\epsilon = 0$ yields perfect privacy (but cannot be obtained with non-trivial accuracy); $\epsilon \to +\infty$ makes the privacy constraint vacuous and leads to no privacy. 

\paragraph{Sensitivity and DP primitives.} The privacy guarantee obtained by an agent depends on how much the output of the computation can change in this agent's data. Intuitively, the more a computation depends on an agent's data (the higher the sensitivity), the harder privacy is to obtain. Formally:

\begin{definition}[Sensitivity]
For a real-valued query $q(\mathbf{x})$, the sensitivity for agent $i$ is
\[
\Delta_i q
\;=\;
\max_{\mathbf{x},\mathbf{x}' \text{ neighbors wrt } i}
|q(\mathbf{x}) - q(\mathbf{x}')|.
\]
\end{definition}

This sensitivity formally captures the maximum influence that agent $i$’s data can have on the output of the computation. Given this sensitivity, a standard method for achieving differential privacy for numerical queries is to add Laplace noise calibrated to the sensitivity of the query.

\begin{definition}[Laplace Mechanism]
Given a query $q(\mathbf{x})$ and parameter $\eta>0$, the Laplace mechanism outputs
$
\mathcal{M}(\mathbf{x}) = q(\mathbf{x}) + Z,~\text{where }~Z \sim \mathrm{Lap}(\eta).
$
\end{definition}
The Laplace mechanism is $\epsilon$-differentially private with respect to agent $i$ when $\eta = \Delta_i q / \epsilon$. Equivalently, for fixed noise scale $\eta$, the privacy level obtained by agent $i$ is
$
\epsilon_i = \frac{\Delta_i q}{\eta}.
$

\paragraph{Privacy and accuracy.}
Because the Laplace mechanism is unbiased, the accuracy of a private estimator can be measured through its variance. If $Z \sim \mathrm{Lap}(\eta)$, then $\mathrm{Var}(Z) = 2\eta^2$. Thus, for an unbiased estimator $\hat{q} = q(\mathbf{x}) + Z$, we have
$
\mathrm{Var}(\hat{q}) = \mathrm{Var}(q(\mathbf{x})) + 2\eta^2.
$

\paragraph{Post-processing immunity.} A key property of differential privacy is that privacy guarantees are preserved under arbitrary data-independent post-processing. I.e., if a data product is released differentially privately, no post-hoc attack 
can break the DP guarantee. Formally:

\begin{theorem-nono}[Post-processing \citep{dwork2014algorithmic}]
If a mechanism $\mathcal{M}$ is $\epsilon$-differentially private with respect to agent $i$, then for any (possibly randomized) function $f$ independent of the data, the composed mechanism $f(\mathcal{M}(\cdot))$ is also $\epsilon$-differentially private with respect to agent $i$.
\end{theorem-nono}
Intuitively, one can see differential privacy as an \emph{information-theoretic} measure of what information is encoded in the data; post-processing a differentially private statistic or model cannot increase the amount of information encoded in the release.

\section{Model}\label{sec:model}

We first introduce the major components of our model: 

\paragraph{Players.} We consider a set of $n$ players indexed by $i \in \{1,\dots,n\}$. Each player $i$ holds a private data point $x_i \in [0,1]$, drawn independently from a distribution $\mathcal{D}$ with unknown mean but \emph{known} variance $\sigma^2$. (Note that the assumption of known variance is the same that is made in the most closely related works~\citep{cummings2015accuracy,fallah2022bridging,satml2023,fallah2023optimal}. Other related works use similar assumptions that control the variance through bounded data or sub-gaussianity). Then, the agents' goal is to estimate the mean through data sharing. Estimation problems (specifically, mean estimation) have been extensively studied in the literature that deals with data acquisition under privacy constraints~\citep{ghosh_selling_privacy,ghosh2013privacy,buyingwithoutverification,roth2012conducting,satml2023,cummings2015accuracy}. This particular variant (with known variance) also arises in many real-world contexts where the variance is well-understood from historical data, but the mean evolves over time. Each player's data is sensitive and must be protected prior to any sharing. Each player $i$ is characterized by a single-dimensional privacy cost parameter $c_i > 0$, which captures her privacy attitude and impacts how costly it is for that player to share her data. In our setting, players are allowed to have \textit{heterogeneous privacy attitudes}. We assume that all $c_i$'s come from the bounded interval $[c_{min}, c_{max}]$ with $c_{min} > 0$. Without loss of generality, the cost parameter $c_i$'s are ordered, i.e., $c_1 \leq c_2 \leq ..... \leq c_n$.

\paragraph{Coalitions and $\epsilon$-profiles.} We consider a setting with at most one coalition. Let $S \subseteq \{1,\dots,n\}$ denote the set of participating players, or the \emph{coalition}. Players outside the coalition do not share their data and do not benefit from the joint computation. We do not consider the formation of multiple coalitions; players decide only whether or not to join the single coalition $S$ (this can be seen as a non-cooperative game theory problem where players simultaneously decide whether to participate, without being able to coordinate with others). Any non-trivial coalition $S$ is of size $|S| \geq 2$. When no data sharing occurs between the players, the coalition is \textit{empty}, hence trivial, and has size $0$. Unless otherwise specified, a `coalition' in our setting refers to a non-trivial one.

If player $i$ participates, i.e., if $i \in S$, she locally perturbs her data using the Laplace mechanism, and obtains a noisy data point $\tilde{x}_i$, as follows:
\[
\tilde{x}_i = x_i + Z_i,~~\text{where}~Z_i\sim\text{Lap} \left(\frac{1}{\epsilon_i}\right),
\]
where $\epsilon_i>0$ is the privacy level of player $i$. The $\epsilon_i$ could either be chosen for the player by the trusted central designer or is chosen endogenously by the player themselves (we will elaborate on this aspect later). We write $\vec{\epsilon}(S)$ the vector of $\epsilon_i$'s for players $i \in S$, and call this vector the \emph{$\epsilon$-profile of a coalition}.

\paragraph{Statistical computation:} The coalition aims to estimate the mean of distribution $\mathcal{D}$. Given a coalition $S$ and a privacy profile $\vec \epsilon=(\epsilon_i)_{i\in S}$, the estimator is given by 
\[
\hat \mu(S,\vec{\epsilon}) \;=\; \frac{1}{|S|}\sum_{i\in S} (x_i + Z_i).
\]
Conditioned on $S$ and $\vec \epsilon$, the resulting variance of the estimator is then given by:
\begin{align}\label{eq:variance}
\text{Var}(\hat \mu (S, \vec \epsilon)) =  \frac{\sigma^2}{|S|} + \frac{2}{|S|^2} \sum_{i\in S}\frac{1}{\epsilon_i^2}.
\end{align}
Note that this variance is \emph{decreasing} in each $\epsilon_i$: the less noise is added for privacy, the more accurate the final computation is. In turn, privacy does not come for free, and the decreased accuracy at increased privacy levels drives the main tension behind our paper. Players who are not part of the coalition, however, only see their own data point. The estimate for a player $i \notin S$, i.e., who does not join the coalition, is $\hat \mu_i = x_i$, and has variance $\sigma^2$.

\paragraph{Player ``burden'':} Each player experiences a ``burden'', which depends on whether they participate in the coalition; this burden results from two types of effects. The smaller the burden, the bigger the benefit to the player\footnote{Minimizing a player's burden is equivalent to maximizing her utility; we present all results in terms of social cost for simplicity of exposition, and note that this is without loss of generality.}.  We assume players' burdens are additive in the accuracy and privacy costs defined below, as is common in the literature ~\citep{fallah2023optimal,satml2023,liao2022privacy,acemoglu2022,nissim2012privacy}:
\begin{enumerate}
\item \textit{Accuracy cost:} \emph{All players}, regardless of whether they participate, incur an \emph{accuracy cost}---the better the accuracy of their own estimate, the lower this cost. For players in the coalition, this cost is just given by the variance of the estimator (note that higher variance means less accuracy, hence higher cost and lower utility), i.e.,
\[
\frac{\sigma^2}{|S|} + \frac{2}{|S|^2} \sum_{i\in S}\frac{1}{\epsilon_i^2}.
\]
This cost is identical across all players in $S$. On the flip side, players not in $S$ all get the same cost of $\sigma^2$, representing the (in)accuracy cost of making a prediction with only their own data. 

\item \textit{Privacy cost:} \emph{Only players who join the coalition} incur a \emph{privacy cost}. This privacy cost comes from other players seeing their data, or a version or transcript of their data (which could be the final statistic itself) and depends on $\epsilon_i$. I.e., player $i\in S$ additionally incurs a privacy cost of the form:
\[
c_i\cdot  f(|S|)\cdot  \epsilon_i.
\]
This privacy cost is increasing and linear in $\epsilon_i$, as introduced by \citet{ghosh_selling_privacy} and then used in much of the literature cited above that studies data acquisition with differential privacy, including some of the most related works~\citep{fallah2022bridging,fallah2023optimal,satml2023}).  
Note that a larger $\epsilon_i$ means less noise hence less privacy (as per Preliminary Section~\ref{sec:dp_prelims}), hence a higher privacy cost. $f(|S|)$ is a nonnegative function capturing how the players' privacy risks scale with the size of the coalition. We introduce this to ensure that they cover different common types of privacy observation models and attacks. 
In particular, the three regimes below are of particular interest:
\begin{itemize}  
\item $f(|S|) = O(1)$ corresponds to the standard \emph{local information-theoretic} differential privacy model. Because differential privacy satisfies post-processing immunity (see Section~\ref{sec:dp_prelims}), no amount of additional computation on someone's data can weaken their privacy guarantee (including a large number of players pooling their resources together). This is also how privacy accounting works in the local differential privacy model. 
\item $f(|S|)$ decreasing in $S$ often arises in \emph{differentially private federated learning} settings, where noise is often added locally, then the data is aggregated and the statistic computed securely---often through SecAgg protocols~\citep{bonawitz2016practical}---, but the privacy guarantee is evaluated centrally\footnote{Theoretically, it is better to add noise centrally. However, in practice, privacy noise is added locally for engineering (in particular, dropouts and colluding players can break the guarantees given by central addition of noise)~\citep{vithana2025correlated}, communication complexity (randomized compression schemes often provide some privacy locally)~\citep{kairouz_compression,youn2023randomized}, and trust (players may not trust a centralized designer to properly privatize their data) reasons.}, at the level of the computed model---see for example~\citet{kairouz2021distributed}.  There, each player in the computation gets a privacy amplification factor of $1/|S|$ for ``free'' from measuring privacy \emph{after} the data is aggregated, compared to their local noise. 
\item $f(|S|)$ increasing in $|S|$ can happen in \emph{fully adversarial} or \emph{privacy leakage} models. In this model, each user in the coalition is an adversary; one incurs harm for every other user that infers information about them. For example: the more players know something about a player's valuation in an auction, the more of them can strategize against her. The privacy leakage model is that of an outside attacker, where more people having access to one's data increases the amount or probability of leakage\footnote{E.g., it is more likely than an adversary can run a successful attack that recovers part of someone's data; though, the maximum amount of information about the data that can be leaked in case of any number of successful attacks is bounded by the differential privacy and post-processing guarantee.}.
\end{itemize}
In order to simplify the exposition and for analytical tractability, we will consider the following functional form of $f(\cdot)$: $f(|S|) = |S|^{\alpha}$ where $\alpha$ can vary in the interval $[-1, 1]$. Note that our function $f(\cdot)$ smoothly interpolates over all of the motivating examples mentioned above. Note that all players in the game have the same exogenous $\alpha$; this $\alpha$ is a product of the environment, in particular of the way the data is shared, and the privacy observation model. 
\end{enumerate}

As explained earlier, the \emph{burden} of a player is additive in the accuracy and privacy costs. So, we define a player's burden formally as follows:

\begin{definition}[Player's burden]
Consider a coalition $S$ with $\epsilon$-profile given by $\vec \epsilon(S)$. Then, the burden $B_i(c_i,\vec \epsilon(S) )$ of a player $i$ with cost $c_i$ is defined as follows:
\begin{itemize}
    \item \emph{Players inside the coalition:} If $i \in S$, $B_i(c_i,\vec \epsilon ) = \frac{\sigma^2}{|S|} + \frac{2}{|S|^2} \sum_{i\in S}\frac{1}{\epsilon_i^2} + c_i f(|S|) \epsilon_i$.
    \item \emph{Players outside the coalition:} If $i \notin S$, $B_i(c_i,\vec \epsilon ) = \sigma^2$.
\end{itemize}
\end{definition}

\paragraph{Social Cost:} Finally, we define our concept of \emph{social cost}, which measures the total cumulative burden of the population of $n$ players. We will henceforth refer to social cost, along with variance of the shared estimator, as our measures of ``performance'' of a coalition. Formally:
\begin{definition}
Given a coalition $S$ of size $|S| \geq 2$ with $\epsilon$-profile $\vec{\epsilon}$, the social cost is given by 
\begin{align*}
\textsf{SC}(S, \vec \epsilon) 
\triangleq \sum_{i \in [n]} B_i(c_i,\vec{\epsilon}) 
= \sum_{j \in S}\left(\frac{\sigma^2}{|S|} + \frac{2}{|S|^2} \sum_{i \in S} \frac{1}{\epsilon_i^2} + c_j f(|S|)\epsilon_j\right) + (n-|S|) \sigma^2 
\end{align*}
Further, if no coalition forms (i.e., $S = \emptyset$), then we have: $SC(\emptyset) = n \sigma^2$.
\end{definition}
Importantly, the social cost includes the burdens of not only players in the coalition $S$, but also of the $n-|S|$ players \emph{outside the coalition}, given by $(n-|S|) \sigma^2$.\\

\paragraph{A simple illustration.} We now provide a toy example to highlight how the degree of decentralization, and how much decision power is left to the players, affects the outcomes.

\begin{example}
Consider a coalition of \(k=3\) players with identical privacy costs \(c_i=1\), suppose that \(f(|S|)=1\), and fix \(\sigma^2=1\). A simple calculation shows that if the three players choose their privacy levels selfishly, then each chooses \(\epsilon_i^d=(4/9)^{1/3}\), whereas a centralized designer minimizing social cost for the same coalition chooses \(\epsilon_i^c=(4/3)^{1/3}\). Thus, even holding the coalition fixed, the decentralized players choose more stringent privacy levels, with \(\epsilon_i^d=3^{-1/3}\epsilon_i^c\). The resulting estimator variance under decentralization is \(1/3+2/(3(\epsilon_i^d)^2)\approx 1.48\), while under the centralized choice it is \(1/3+2/(3(\epsilon_i^c)^2)\approx 0.89\). Similarly, the social cost is \(SC_d\approx 6.72\), compared to \(SC_c\approx 5.95\). 

The key tension here is that each player only anticipates the impact of their \emph{own} decision on the accuracy of the final estimator. As such, they under-estimate their joint impact on the variance, and overfocus on minimizing their privacy costs. As a result, decentralized privacy choices are too stringent, leading to both higher estimator variance and higher social cost. We will formalize all these ideas later. 
\end{example}

\subsection{Coalition Formation under Full Decentralization} 
We first consider a mechanism where players have \textit{complete autonomy}. Each player decides: 
\begin{enumerate}
    \item whether to participate in the coalition; and,
    \item if participating, what \textit{endogenous} level of noise $\epsilon_i$ to add to their private data point before sharing it with the coalition. We note that local addition of noise is a standard and minimal trust model in the local differential privacy literature~\citep{duchi2013local}. In federated learning systems, faithful execution of the chosen privacy mechanism can be enforced via cryptographic verification schemes such as zero-knowledge proofs of privacy~\citep{biswas2022verifiable,narayan2015verifiable}.
\end{enumerate}
Since players are strategic and make decisions in a decentralized manner, the only coalitions that will exist are those that are \textit{stable}.  
Note that we operate under the assumption that the cost profile $\vec{c}$ is common knowledge. All players are assumed to be rational who can anticipate each other's best responses given the cost profile. Further, we assume common knowledge of rationality, as is standard in non-cooperative game theory~\citep{Nash51,aumann76}.

\paragraph{Stability of Coalitions:} Before we talk about stability at the coalition level, we need to understand how player $i \in S$ chooses their privacy parameter $\epsilon_i$, if a coalition $S$ were fixed (with respect to entry and exit of players). For each player $i \in S$, we define: 
\[
    \epsilon_i^*(S) = \arg \min_{\epsilon_i} B_i(c_i, (\epsilon_i, \vec \epsilon_{-i}) ).
\]
Recall that each player's burden $B_i(\cdot)$ is additively separable in the $\epsilon$ choices of different players and convex in $\epsilon_i$, which means that $\epsilon_i^*(S)$ exists, is unique and is the same as player $i$ selfishly optimizing their own burden (without worrying about other players' choices of privacy levels). We write $\vec \epsilon^*(S)$ as the vector of these player-wise $\epsilon_i^*(S)$'s.  
Further, note that given a fixed coalition $S$, $\vec \epsilon^*(S)$ constitutes the unique Nash equilibrium of the sub-game where burden-minimizing players in $S$ are trying to choose their own privacy levels. From now on, we assume that in the fully decentralized setting, all players \textit{when acting autonomously} try to minimize their cost/burden. So, once $S$ is fixed, any player $i \in S$ always chooses $\epsilon_i = \epsilon_i^*(S)$ (we derive this in closed form in Section~\ref{sec:autonomous} Lemma~\ref{lem:autonomous_eps}). 

We can now define stability at the coalition level.  We introduce two definitions of coalition stability, both of which rely on a similar idea: no external player wants to enter the coalition and no existing player wants to leave. Our first definition is a natural extension of standard Nash stability, to account for the fact that players not only make decisions about joining or leaving the coalition, \emph{but also on their choice of privacy level}. In particular, it refers to stability in the sense that no player $j$ wants to unilaterally change their participation decision or their own choice of privacy level, assuming that all other player decisions are fixed.

\begin{definition}[Nash-stable coalition]
A coalition $S$ with privacy profile $\vec{\epsilon}^*(S)$ is Nash-stable if and only if: 
\begin{itemize}
\item For all $j \in S$, $\frac{\sigma^2}{|S|} + \frac{2}{|S|^2} \left(\sum_{i\in S}\frac{1}{(\epsilon_i^*(S))^2}\right) + c_j f(|S|) \epsilon_j^*(S) \leq \sigma^2$; \hfill \emph{(No player wants to leave)}
\item For all $j \notin S$, \hfill \emph{(No outside player wants to join)}
\\$\min_{\epsilon_j} \frac{\sigma^2}{|S|+1} + \frac{2}{(|S|+1)^2} \left(\frac{1}{\epsilon_j^2} + \sum_{i\in S}\frac{1}{(\epsilon_i^*(S))^2}\right) + c_j f(|S|+1) \epsilon_j > \sigma^2$. 
\end{itemize}
\end{definition}
\noindent
The first condition encodes that a player currently in the coalition cannot leave and get a lower cost. The second condition encodes that a player outside the coalition does not get a better cost by entering the coalition. Importantly, because the standard definition of Nash-stability considers unilateral deviations, a player $j \notin S$ makes an entry decision \emph{assuming that the $\vec{\epsilon}$ of players currently inside the coalition is fixed}.

Beyond the standard Nash definition, it is natural to ask in our setting: ``what happens to the privacy choices of players inside the coalition when a new player attempts to join?'' In principle, sophisticated players should be able to anticipate the change between $\vec \epsilon^*(S)$ and $\vec \epsilon^*(S \cup \{j\})$ when an outside player $j$ enters the coalition $S$. There are two things to consider here: i) whether the entrant $j$ still wants to join the coalition (taking into account how her entry will change the coalition's privacy parameters)---in particular, a player may be able to enter given the current profile of $\epsilon$'s, but may no longer find it favorable to join if other players adjusted their $\epsilon$; ii) whether all current players $i \in S$ would want to remain in the coalition after $j$'s entry (accounting for how it would affect everyone's privacy parameters and their own costs). If the augmented coalition $S \cup \{j\}$ can exist without anyone wanting to leave, player $j$ would be allowed to enter $S$ --- we call this \textit{a valid entry into $S$}. However, if there exists even a single player $i \in S \cup \{j\}$ who would be unhappy with the arrangement (because they face a higher cost than the outside option) and would want to leave, we assume that the entry of player $j$ will be vetoed. 

We thus define a new notion of stability as follows: a coalition is stable iff i) no player wants to leave the coalition; and ii) no player from outside can \textit{validly} enter it. We call such coalitions \emph{equilibria that are robust to valid entry}, or \emph{robust equilibria} in short. Note that the robust definition is identical to the Nash definition when it comes to players in the coalition $S$. However, a player outside the coalition now faces a harsher condition to join, which helps the coalition remain stable for longer. Formally, 

\begin{definition}[Equilibrium coalition robust to valid entry]\label{def:valid_entry}
Given a coalition $S$, a player $j \notin S$ can make a \textbf{valid entry} into $S$ if and only if, for all players $l \in S \cup j$,
\[
\frac{\sigma^2}{|S|+1} + \frac{2}{(|S|+1)^2} \left(\sum_{i \in S \cup {j}}\frac{1}{(\epsilon_i^*(S \cup \{j\}))^2 }\right) + c_l f(|S| + 1) \epsilon_l^*(S \cup \{j\}) \leq \sigma^2.
\]

A coalition $S$ with privacy profile $\vec{\epsilon}^*(S)$ is robust-stable if and only if:
\begin{itemize} 
\item For all $j \in S$,  $\frac{\sigma^2}{|S|} + \frac{2}{|S|^2}\left( \sum_{i\in S}\frac{1}{(\epsilon_i^*(S))^2}\right) + c_j f(|S|) \epsilon_j^*(S) \leq \sigma^2$;  \hfill \emph{(No player wants to leave)}
\item For all $j \notin S$, $\exists$ $l \in S \cup \{j\}$ such that:  \hfill \text{\emph{(No player $j \notin S$ can validly enter)}}
\\$\frac{\sigma^2}{|S|+1} + \frac{2}{(|S|+1)^2} \left(\sum_{i\in S \cup {j}}\frac{1}{(\epsilon_i^*(S \cup \{j\}))^2}\right) + c_l f(|S|+1) \epsilon_l^*(S \cup \{j\}) > \sigma^2$.  
\end{itemize}
\end{definition}

We can formalize the relationship between our two notions of stability as follows:

\begin{claim}\label{clm:Nash_stronger}
Given $(\vec c, \sigma^2)$ and $\alpha \in [-1, 1]$, any $S \subseteq [n]$ which is a Nash-stable coalition must also be a robust-stable coalition. 
\end{claim}
The above relationship is consequential as many properties of the equilibria established under the Nash definition immediately extend to the robust definition. 

\begin{remark}
Note that beyond the two notions of stability we are using, there are several other possible ways to think of stability in the context of coalitions. For example, 
\emph{individual stability} requires that the entry of a new player is vetoed if any existing player in the coalition risks facing a higher burden than the existing arrangement. 
Similarly, there are other notions of stability studied in the cooperative game theory literature like \textit{core stability}~\citep{gillies1953some,shapley1965balanced,aumann1961core} where the arrangement needs to be stable against any subset of current participants jointly breaking out and forming their own coalition.  

However, we argue that our main insights are not driven by the specific choice of notion of stability. In fact, we will show subsequently that the major cause of inefficiency in the fully decentralized setting is due to players choosing their individual $\epsilon$'s sub-optimally; this is evidenced by the fact that partial decentralization (Section~\ref{sec:partial-def}), where all agents must use a single, externally imposed $\epsilon$, recovers the social cost and estimator variance of the fully centralized baseline mechanism (Section~\ref{sec:baseline-def}). We also show that players' choices of privacy levels under full decentralization are completely unaffected by the definition of stability, and differ in order from those chosen by the designer in the partially decentralized (Section~\ref{sec:partial}) or fully centralized settings (Section~\ref{sec:baseline}). 
\end{remark}

\subsection{An Intermediate Mechanism: Coalition Formation under Partial Decentralization}\label{sec:partial-def}
We now introduce a partially decentralized data sharing mechanism. In this mechanism, there is a central designer who chooses a fixed privacy level $\epsilon > 0$ at the beginning. Players still retain the agency to make individual participation decisions like the fully decentralized mechanism described earlier, however, anybody who chooses to participate receives the same privacy level of $\epsilon$. In this sense, our mechanism is \textit{partially} decentralized, given the centralized privacy level.

\paragraph{Stability of Coalitions:} 
In this mechanism, the notion of stability is greatly simplified because players can only make participation decisions: informally, a coalition is considered \textit{stable} if no player inside the coalition has an incentive to leave unilaterally and no player outside the coalition has an incentive to join unilaterally. As a result, the standard notion of Nash-stability applies directly~\citep{bogomolnaia2002stability}. Formally,

\begin{definition}[Nash-stable coalition at fixed privacy level $\epsilon$] 
Given a fixed privacy level $\epsilon$, a coalition $S$ is Nash-stable if and only if: 
\begin{itemize}
    \item For all $j \in S$, $\frac{1}{|S|}\left(\sigma^2 + \frac{2}{\epsilon^2} \right) + c_j f(|S|)\epsilon \leq \sigma^2$; \hfill \emph{(No player wants to leave)}
    \item For all $j \notin S$, $\frac{1}{|S|+1}\left(\sigma^2 + \frac{2}{\epsilon^2} \right) + c_j f(|S|+1)\epsilon > \sigma^2$. \hfill \emph{(No outside player wants to join)} 
\end{itemize}
\end{definition}

\subsection{The Baseline Mechanism: Coalition Formation under Full Centralization}\label{sec:baseline-def}
We finally introduce the baseline mechanism against which we will measure the performance of both of our decentralized mechanisms. In this setting, no player has decision-making authority. Instead, all decisions are made by the central designer who chooses: 
\begin{enumerate}
    \item which players to include in the coalition; and
    \item what privacy level $\epsilon_i$ to assign to each participant in the coalition. 
\end{enumerate}

It should be clear that the centralized mechanism always achieves the \textit{best outcome} in terms of minimizing social cost compared to all other mechanisms. We can then define a measure of efficiency called the \textit{Price of Stability} (\textsf{PoS}) with respect to performance metric $X$ (like social cost or estimator variance), using the centralized mechanism as a baseline: 
\[
    \textsf{PoS}(X) = \frac{\text{Value($X$) Achieved by Optimal Coalition under Decentralization at $(\vec c, \sigma^2)$} }{ \text{Value($X$) Achieved by Optimal Coalition under Centralization at $(\vec c, \sigma^2)$} }.
\]
Note here that the above definition applies to both our fully decentralized and partially decentralized mechanisms. 
Additionally, the optimal coalition under centralization versus decentralization can be very different: the \textit{optimal} coalition in the centralized setting is the one which minimizes social cost \textit{across all possible coalitions}, while the \textit{optimal} coalition in the decentralized setting (full or partial) is the one which minimizes social cost only \textit{out of those coalitions which are stable}. 
In the above definition, we use the convention that in case no stable coalitions exist under decentralization, we default to the setting where there is no data sharing between players. 
We consider two natural choices for the performance metric $X$: 
\begin{itemize}
\item the social cost $\textsf{SC}$;
\item the variance $\textsf{Var}$ of the estimator computed by the coalition. 
\end{itemize}
Note that, by definition, $\textsf{PoS}(\textsf{SC}) \geq 1$. However, this may not be true for $\textsf{PoS}(\textsf{Var})$: the optimal coalition in terms of social cost may not always be the optimal coalition in terms of variance.

\section{What Happens under Full Centralization?}\label{sec:baseline}

In this section, we explore the baseline model (with full centralization) in greater detail, investigating how a central designer, interested in minimizing the social cost, would choose to assign privacy levels to players if they had complete authority. We are also interested in understanding how these decisions would affect the composition and performance of the social-cost-minimizing coalition across different regimes of $\alpha$. This performance is again defined in terms of both social cost and variance.  

In order to obtain the social cost minimizing coalition in the centralized setting, the designer needs to jointly optimize over the size and composition of the coalition and the privacy levels of the players.
To do so, we first solve the lower-level problem of deciding, assuming we know the coalition size $k$, the optimal set of participants such that $|S| = k$ as well as the corresponding optimal privacy levels; then, we solve the higher-level problem of optimizing over the coalition size $k$.
\subsection{Optimal Decisions at a Fixed Coalition Size} 
Our first main result solves the lower-level optimization problem for the designer. We define the \emph{optimal coalition at size $k$} as the coalition that incurs the lowest social cost among all coalitions of size $k$.  
\begin{lemma}\label{lem:centralized}
Given ($\vec c, \sigma^2$) and integer $k$ ($n \geq k \geq 2$),  
\begin{itemize}
    \item the optimal coalition at size $k$ is downward-closed, i.e., it includes the $k$ players with the first $k$ smallest costs (we denote this coalition by $S_k$);
    \item the designer's optimal choice of privacy level for each player $i \in S_k$ is given by $\epsilon_i = \left(\frac{4}{k^{(1+\alpha)} c_i} \right)^{1/3}$; 
    \item the estimator variance at the optimal coalition at size $k$ is given by:
    \begin{align}\label{eq:varbest_c}
             \textsf{Var}_c(\hat \mu~|~k) = \frac{1}{k}\left[\sigma^2 + \frac{ k^{2(\alpha+1)/3}}{2^{1/3}} \left(\sum_{i \in S_k}c_i^{2/3}/k \right) \right]; \quad \text{and}
    \end{align}
    \item the social cost achieved at the optimal coalition at size $k$ (also the best achievable social cost at size $k$) is given by
    \begin{align}\label{eq:scbest_c}
        \textsf{SC}_{c}(k) = (n+1)\sigma^2 - \left[ k\sigma^2 - \left(\frac{3}{2^{1/3}} \right)k^{2(\alpha +1)/3}\left(\sum_{i \in S_k}c_i^{2/3}/k \right) \right].
    \end{align}
\end{itemize}
\end{lemma}
We make the following observations: firstly, the above result tells us a key structural property of the optimal coalition at size $k$---this reduces the designer's higher-level optimization from a combinatorial problem to one over a single-dimensional variable (the size of the coalition).   
Secondly, even when minimizing social cost, it is in the designer's best interest to respect the heterogenous privacy attitudes of individual players, i.e., participants with a higher privacy sensitivity $c_i$ obtain a more stringent privacy level (smaller $\epsilon_i$). Further, all participants obtain \emph{more stringent} privacy levels as the coalition size increases for any $\alpha \in (-1, 1]$---this is perhaps counterintuitive in the case where $\alpha < 0$ where having more agents in the coalition leads to improved privacy costs. The best $\epsilon$ choices become independent of coalition size only at $\alpha = -1$. .  

\subsection{What is the Optimal Coalition Size, and What are its Implications for Variance and Social Cost?}
We now solve the higher-level problem for the designer: finding the optimal coalition size $k^*$.  
\begin{align}\label{opt:kstar}
        k^* = \begin{cases}
        &\arg \min_{k \in \{2,3...n\}} \textsf{SC}_c(k), \quad \text{if } \min_{k \in \{2,3...n\}}\textsf{SC}_c(k) < n\sigma^2\\
        &0, \quad \text{o/w.}
        \end{cases}
\end{align}

\noindent 
The value of $k^*$ depends greatly on the regime of $\alpha$ we are in and the parameters of the problem like the number of players $n$, the cost profile $\vec c$ and the variance $\sigma^2$. For any arbitrary problem instance, $k^*$ can be found efficiently by solving the one-dimensional optimization problem (Program~\eqref{opt:kstar}). Note that $k^*$ can be equal to $0$ if no coalition of size $\geq 2$ offers an improvement over the trivial setting of the empty coalition. We are particularly interested about the value of $k^*$ where the number of players $n$ is large.

\begin{thm}\label{thm:kstar_c}
When the number of players $n$ is sufficiently large,
\begin{enumerate}[label=\roman*)]
    \item For $\alpha \in \left[-1, \frac{1}{2}\right)$, we have $k^* = n$, i.e., the grand coalition is optimal. This achieves $\textsf{Var}_c(\hat \mu~|~k^*) = \Theta(n^{(2\alpha-1)/3})$ and $\textsf{SC}_c(k^*) = \Theta(n^{2(\alpha+1)/3})$.
    \item For $\alpha = \frac{1}{2}$, $k^*$ can take any value in $\{0,2,3,...n\}$. In all cases, $\textsf{Var}_c(\hat \mu~|~k^*) = \Theta(1)$ and $\textsf{SC}_c(k^*) = \Theta(n)$.
    \item Finally, for $\alpha \in \left(\frac{1}{2}, 1\right]$, we either have $k^* = 0$ or $k^* = \Theta(1)$. In both cases, we obtain $\textsf{Var}_c(\hat \mu~|~k^*) = \Theta(1)$ and $\textsf{SC}_c(k^*) = \Theta(n)$.  
\end{enumerate}
\end{thm}
\begin{proofsketch}
We now provide some intuition. From Equation~\eqref{eq:scbest_c}, we note that the social cost is minimized when the function $g(k) = k\sigma^2 - \left(\frac{3}{2^{1/3}}\right)\left(\sum_{i \in S_k}c_i^{2/3}/k\right) k^{2(\alpha+1)/3}$ is maximized. $g(k)$ consists of a first term that increases linearly in $k$ and a second term which decreases in $k^{2(\alpha+1)/3}$---note that by assumption on the support of the costs, $\sum_{i \in S_k}c_i^{2/3}/k = \Theta(1)$. 

The value of $k^*$ will therefore be determined by the interplay between these two terms. When $\alpha < \frac{1}{2}$, the second term is sub-linear which means that for sufficiently large $n$, the first term would always dominate, implying that $g(k)$ will be maximized at $n$. Therefore $k^* = n$ for $\alpha \in \left[-1, \frac{1}{2}\right)$. $\alpha = \frac{1}{2}$ is a special case where both terms are linear in $k$ and the location of the maximizer is determined by the net sign. 
Finally, for $\alpha > \frac{1}{2}$, the second term begins to dominate and the function $g(k)$ either becomes monotonically decreasing or concave, both of which lead to a unique maximizer independent of $n$. This implies that $k^* = 0$ or $\Theta(1)$. The values of $k^*$ in the different regimes then drive the orders of the optimal social cost and the optimal variance.  
\end{proofsketch}
\noindent 
While Theorem~\ref{thm:kstar_c} completely characterizes the optimal coalition characteristics for all regimes of $\alpha$ for large $n$, we specifically want to draw the reader's attention to the following key takeaway: \textit{$\alpha \in \left[-1, \frac{1}{2} \right)$ is the only regime where the optimal accuracy of the estimator achieves non-trivial improvement over the outside option $\sigma^2$ (which is $\Theta(1)$) and the optimal social cost grows sublinearly in $n$}. This indicates that the rate at which privacy costs grow in the size of the coalition may have implications for how effective data sharing is. When costs grow rapidly in coalition size, the central designer is required to add more noise to each data point (by offering smaller $\epsilon$'s), thereby making the accuracy of the computation worse. This further dissuades the designer from including players with high privacy sensitivities ($c_i$'s) into the coalition explaining the formation of coalitions of size $\Theta(1)$ or no coalition at all.

\section{What Happens under Full Decentralization?}\label{sec:autonomous}

In this section, we turn our attention back to coalition formation through players making self-interested decisions in a fully decentralized way. We aim to answer three key questions.
\begin{enumerate}
\item How do participants in a stable coalition $S$ choose their individual privacy levels? What would be the social cost and estimator variance achieved by $S$ in this case? 
\item What can we say about the existence of stable coalitions and what would be some of their properties? (Section~\ref{subsec:stability})
\item For different regimes of $\alpha$, what are the characteristics of the optimal (social cost minimizing) 
stable coalition? (Section~\ref{subsec:optimal_eq})
\end{enumerate} 
\noindent 
First, we investigate how players make decisions when they are a part of an fixed coalition and how it affects the estimator variance and the social cost at said coalition. Lemma \ref{lem:autonomous_eps} examines a single coalition and describes the privacy level that each player autonomously chooses for themselves when optimizing their own burden, and consequentially, the variance and social cost associated with this coalition. Note that Lemma \ref{lem:autonomous_eps} would  apply to any coalition (hence to both Nash and robust stable coalitions). The specific notion of stability affects \emph{which} agents join the coalition, not their privacy level \emph{conditioned} on a specific coalition (Section~\ref{sec:model}). Once a player is part of the coalition, they pick the $\epsilon$ level that unilaterally minimizes their social burden.

\begin{lemma}\label{lem:autonomous_eps}
Given $(\vec c, \sigma^2)$ and coalition $S \subseteq [n]$ with $k \triangleq |S| \geq 2$ with each player in $S$ choosing their own privacy level autonomously, 
\begin{itemize}
\item for any player $i \in S$, their optimal choice of privacy level is given by $\epsilon_i^* = \left(\frac{4}{k^{(2+\alpha)} c_i} \right)^{1/3}$;
\item the estimator variance achieved by pooling data in coalition $S$ is given by: 
\begin{align}\label{eq:var_d}
      \textsf{Var}_{d}(\hat \mu~|~S) = \frac{1}{k}\left[ \sigma^2 + \frac{ k^{2(\alpha+2)/3} }{2^{1/3}}\left(\sum_{i \in S}c_i^{2/3}/k \right) \right]; \quad \text{and}
\end{align}
\item at coalition $S$, the social cost is given by:
\begin{align}\label{eq:sc_d}
     \textsf{SC}_{d}(S) = (n+1)\sigma^2 - \left[k\sigma^2 -\left(\frac{1}{2^{1/3}} \right) (k+2)k^{(2\alpha+1)/3} \left(\sum_{i \in S} c_i^{2/3}/k\right) \right]; 
\end{align}
\end{itemize}
\end{lemma}
At this juncture, it is important to draw some comparisons with the baseline setting (Lemma~\ref{lem:centralized}). While a participant's optimal choice of privacy level $\epsilon_i^*$ has the same dependence on their privacy cost parameter $c_i$ in both mechanisms, in the decentralized setting, the coalition size exerts a larger influence on $\epsilon_i^*$: specifically, note that the $\epsilon_i^*$ in Lemma \ref{lem:autonomous_eps} is larger than the $\epsilon_i^*$ in Lemma \ref{lem:centralized} by a factor proportional to $k^{1/3}$, the size of the coalition. When players make decisions in a decentralized manner, in particular, \emph{everyone opts for more stringent privacy levels} (smaller $\epsilon$'s). 
This implies a perhaps surprising observation: as we will see later in Section~\ref{sec:pos}, decentralization hurts efficiency, not only in terms of social cost but also \emph{in terms of variance of the shared estimator}.

\subsection{Existence of Stable Coalitions: Some Special Cases}\label{subsec:stability}

\paragraph{Existence of Stable Coalitions: Nash-stability vs Robust-stability.} We first highlight that the standard notion of Nash stability can often be \textit{brittle} in our problem setting. In particular, we can show that the existence of Nash-stable coalitions may be non-monotonic in the variance level $\sigma^2$. I.e., there can be intermediate regimes of variance $\sigma^2$ where no Nash-stable coalitions exist. Our robust-stability definition addresses some of these issues with the Nash definition. We can show that it has nice monotonicity properties in terms of both existence and coalition size --- i.e., if a robust-stable coalition of a certain size $k_1$ exists at variance level $\sigma_1^2$, then for any larger variance $\sigma_2^2 > \sigma_1^2$, we are guaranteed to find a robust-stable coalition of size $k_2 \geq k_1$. We elaborate on these properties with detailed proofs and examples in Appendix~\ref{app:extra}.  

\paragraph{Challenges of Computing Stable Coalitions in General Cases.} 
Before we talk about existence of stable coalitions in special cases, it is also important to note some properties of stable coalitions that make efficient numerical computation very difficult. For example, we can show that there exist problem instances $(\alpha, \vec c, \sigma^2)$ where we have stable coalitions which are non-monotonic in player costs (Lemma~\ref{lem:non_monotone_cost}, Appendix~\ref{app:extra}). i.e., such coalitions include higher cost players at the expense of some lower cost players (which is counterintuitive) and are still stable. This often leads to multiplicity of stable coalitions (which might even be of different sizes) (Lemma~\ref{lem:multiplicity}, Appendix~\ref{app:extra}). This lack of structure at stability and the multiplicity property implies that finding stable coalitions for a given problem instance often requires brute-force enumeration over an exponential ($2^n$) number of options which is obviously inefficient as $n$ grows large. This motivates us to find conditions under which some special coalitions of interest are guaranteed to be stable.  

\paragraph{Special case I: Stability of Grand Coalition.} 
The grand coalition (the coalition that consists of all players) is often a desirable outcome for many data-sharing mechanisms because it indicates the maximum extent of data sharing between players. Claim \ref{clm:high_var} provides sufficient conditions for when the \emph{grand coalition} is a stable outcome of the fully decentralized setting. 

\begin{claim}\label{clm:high_var}
Given cost profile $\vec c$ and variance $\sigma^2$, the grand coalition is (Nash/robust) stable if: 
\begin{itemize}
    \item $2 \leq n \leq \left[ \frac{\sigma^2}{4(c_n^2/2)^{1/3}} \right]^{3/(2\alpha+1)}$, for $\alpha > - \frac{1}{2}$; and 
    \item $n \geq \max \left\{ 2, \left[ \frac{4(c_n^2/2)^{1/3}}{\sigma^2} \right]^{-3/(2\alpha+1)} \right\}$, for $\alpha < - \frac{1}{2}$.
\end{itemize}
\end{claim}
Note that the conditions in Claim \ref{clm:high_var} for the grand coalition to be stable are sufficient but not necessary. Yet, they are useful for understanding situations in which all players can stably pool data together. The only requirement to sustain the grand coalition is that no player has an incentive to leave the coalition: i.e., nobody's burden for participation exceeds the cost of non-participation $\sigma^2$. For the grand coalition, the burden of participation is $\Theta(n^{(2\alpha+1)/3})$. This means that in the regime where $\alpha > -\frac{1}{2}$, large values of $n$ make it difficult to sustain the grand coalition. On the other hand, in the $\alpha < - \frac{1}{2}$ regime, increasing the number of players actually lowers privacy costs and the burden of participation for everyone, maintaining stability of the grand coalition as more and more players keep joining. 

\paragraph{Special Case II: Stable Coalitions When All Players Have Identical Costs.}
Finally, we discuss the setting where all players have identical cost parameters $c$. While this setting is simple, it will be particularly useful for establishing matching lower bounds on the best achievable social cost in Theorem~\ref{thm:opt_eq} (because it is one of the few instances where all the equilibrium coalitions that exist at $(c, \sigma^2)$ can be accounted for, which enables finding the one that minimizes social cost). Again, these results will highlight differences between Nash and robust stability. 

First, Claim \ref{clm:equal_cost_Nash} shows that when players have identical cost parameter $c$, the grand coalition is the only coalition that can be Nash-stable for the entire allowable range of $\alpha$. 
\begin{claim}\label{clm:equal_cost_Nash}
Suppose, we are in a setting where all players have the identical cost parameter $c$. Then, there exists no Nash-stable coalition of intermediate size $k$ ($2 \leq k < n$) for any $\alpha \in [-1, 1]$. 
\end{claim}

By contrast, Claim \ref{clm:equal_cost_robust} shows that under the robust definition, the existence of a stable coalition of size smaller than the grand coalition is possible and depends on $\alpha$: 
\begin{claim}\label{clm:equal_cost_robust}
Consider a setting where all players have the identical cost parameter $c$. Then, at any given variance level $\sigma^2$, 
\begin{itemize}
    \item there exists no robust equilibrium of intermediate size $k$ ($2 \leq k < n$) for any $\alpha \in \left[-1, -\frac{1}{2}\right]$;
    \item there exists at most one robust equilibrium of intermediate size $k$ ($2 \leq k < n$) for $\alpha \in \left(-\frac{1}{2}, 1 \right]$.
\end{itemize}
\end{claim}

\subsection{Optimal Stable Coalition}\label{subsec:optimal_eq}

Now that we have some understanding of coalition stability in the decentralized setting, our goal is to characterize properties of the optimal stable coalition. But given the complexities in both notions of stability explored in the previous section (including non-monotonicity, multiplicity, and non-existence), we highlight that finding the optimal coalition is likely to be difficult. So instead, we focus on the setting when the number of players is sufficiently large, complementing our results on the optimal coalition in the centralized setting (Theorem~\ref{thm:kstar_c}). We put particular emphasis on the existence of the grand coalition. 

\begin{thm}\label{thm:opt_eq}
When the number of players $n$ is sufficiently large, 
\begin{enumerate}[label=\roman*)]
    \item For $\alpha \in \left[-1, -\frac{1}{2}\right)$, the grand coalition is both Nash-stable and robust-stable and achieves $\textsf{Var}_d(\hat \mu~|~S_n) = \Theta(n^{(2\alpha+1)/3})$ and $\textsf{SC}_d(S_n) = \Theta(n^{(2\alpha+4)/3})$. Further, the estimator variance and the social cost at the optimal stable coalition in this regime are also $\Theta(n^{(2\alpha+1)/3})$ and $\Theta(n^{(2\alpha+4)/3})$ respectively.
    \item For $\alpha = -\frac{1}{2}$, the grand coalition may not be stable. In this regime, the estimator variance and social cost that can be achieved by the optimal stable coalition (if it exists) are $\Theta(1)$ and $\Theta(n)$ respectively.   
    \item For $\alpha \in \left(-\frac{1}{2}, 1\right]$, the grand coalition is never stable under any definition of stability. Further, in this regime, the estimator variance and social cost that can be achieved at the optimal stable coalition (if it exists) are $\Theta(1)$ and $\Theta(n)$ respectively.  
\end{enumerate}
\end{thm}

\begin{proofsketch}
We start with the case when $\alpha < -\frac{1}{2}$. To prove our upper bound, we remark by Claim~\ref{clm:high_var} that the grand coalition is both Nash-stable and robust stable. At the grand coalition, we have $\textsf{Var}_d(\hat \mu~|~S_n) = \Theta(n^{(2\alpha+1)/3})$ and $\textsf{SC}_d(S_n) = \Theta(n^{(2\alpha+4)/3})$. 
To prove our lower bounds, we study the special case of instances where all players have identical costs. Using Claims~\ref{clm:equal_cost_Nash} and \ref{clm:equal_cost_robust}, we argue that in this special problem instance, the grand coalition is the unique stable coalition for $\alpha < -\frac{1}{2}$ which implies optimality with variance $\sim \Theta(n^{(2\alpha+1)/3})$ and social cost $\sim \Theta(n^{(2\alpha+4)/3})$. 

The next two regimes are then similar in flavor to the first regime. We start by verifying 
whether the grand coalition for $\alpha \in \left[-\frac{1}{2}, 1\right]$ is stable. We then note that the estimator variance and social cost that can be achieved at the optimal stable coalition (if it exists) are $\mathcal{O}(1)$ and $\mathcal{O}(n)$ respectively, which allows us to obtain our upper bounds. We again match these bounds by constructing specific problem instances using the identical cost setting to complete the proof. 
\end{proofsketch}
Importantly, $\alpha \in \left[-1, -\frac{1}{2}\right)$ is the only regime where decentralization achieves non-trivial improvement in estimator variance, as well as a social cost that grows sub-linearly in $n$. This is in sharp contrast to the centralized setting where non-trivial improvements in variance and social cost can be achieved for a much wider range of $\alpha$ ($\alpha \in \left[-1, \frac{1}{2}\right)$). This shows that decentralization introduces inefficiencies into the mechanism of forming coalitions even under complete information.

\section{The Middle Ground: What Happens under Partial Decentralization?}\label{sec:partial}

Our last mechanism of interest is the intermediate mechanism which has partial decentralization. 
Before moving on to the analysis, we reiterate how the partially decentralized setting differs from the baseline setting of full centralization. While in both mechanisms the centralized designer picks out privacy levels (in the baseline setting, the designer has even more flexibility---they can choose personalized $\epsilon$'s for each participant), the main difference is with respect to \textit{stability}. \textit{Under full centralization, there is no requirement for stability}. We can think of the designer as a benevolent dictator who has full control over the system and can form coalitions whichever way they like. On the other hand, the partially decentralized mechanism obviously has a stability requirement---since players now make participation decisions themselves, they will break out of the coalition if doing so makes them better off. Note that stability does not necessarily align with minimizing the social cost, hence the two settings can genuinely be different: a player may enter or leave the coalition if it is good for them, but their entry or exit decision affects the accuracy (variance) of the estimator obtained by all other players in the coalition, which can, in turn, increase the total social cost. 

\begin{example}
We can illustrate with an example: Suppose, there are $5$ players with costs $\vec c= [0.1, 0.1, 0.1, 0.1, 0.4]$. Additionally, the outside variance is $\sigma^2 = 1$ and we are in the regime where $\alpha = 0$. The centralized designer chooses $\epsilon \approx 1.71$ which is the best fixed $\epsilon$ to minimize social cost for the grand coalition if there are no stability requirements ($\epsilon = \left(\frac{4}{5\times \bar c_5}\right)^{1/3} = 5^{1/3} \approx 1.71$). In this case, the grand coalition $S_5$ also achieves the best social cost (value $\approx 3.05$) compared to all other coalitions (next best $3.36$). However, $S_5$ is not stable at this $\epsilon$. Because player $5$ has a burden of participation given by: $\frac{1}{5}\left(1  + \frac{2}{1.71^2}\right) + 0.4\times 1.71 = 1.02$ which is worse than the outside option. This shows that the designer choosing the $\epsilon$ optimally for social cost does not lead to stability. 
\end{example}

From a technical point of view, our approach to the partially decentralized setting is relatively different, and generally more complex, compared to that of the fully decentralized setting. Under full decentralization, the $\epsilon$ levels were entirely determined by the size of the coalition. Here, the picture is flipped: the coalition is determined by the chosen $\epsilon$. Our approach must therefore characterize, given any coalition $S$, the set $\mathcal{R}_{\epsilon}(S)$ of feasible $\epsilon$'s for which this coalition is stable. We then optimize over $\mathcal{R}_{\epsilon}(S)$ to find the best $\epsilon^*(S)$ (that minimizes social cost) for said $S$.

\begin{lemma}\label{lem:partial_general}
Given $(\vec c, \sigma^2)$ and an integer $k \in \{2,3,...n\}$, under partial decentralization, 
\begin{itemize}
    \item the set of feasible $\epsilon$'s at which a coalition $S$ with $|S| = k$ is Nash-stable is given by: 
    \[
         \mathcal{R}_{\epsilon}(S) = \left\{ \epsilon > 0:  ~\frac{2}{k\epsilon^2} + \frac{(k+1)^{1+\alpha}}{k}\cdot \epsilon \min_{i \notin S}c_i > \sigma^2 \geq \frac{2}{(k-1)\epsilon^2} + \frac{k^{1+\alpha}}{k-1}\cdot \epsilon \max_{i \in S}c_i \right\} \quad (\text{when }k < n);
    \]
    \[
          \mathcal{R}_{\epsilon}(S) = \left\{ \epsilon > 0:  ~\sigma^2 \geq \frac{2}{(k-1)\epsilon^2} + \frac{k^{1+\alpha}}{k-1}\cdot \epsilon \max_{i \in S}c_i \right\} \quad (\text{when }k = n);
    \]
    \item if $S$ is Nash-stable at $\epsilon$, then $S_k$ is also Nash-stable at the same $\epsilon$, i.e., $\mathcal{R}_{\epsilon}(S) \subseteq \mathcal{R}_{\epsilon}(S_k)$; 
    \item the best social cost achievable by a Nash-stable coalition of size $k$ (if it exists) is given by: 
    \[
        \textsf{SC}_{f}(k, \epsilon_{(k)}^*) = (n+1)\sigma^2 - \left[k \sigma^2 - \frac{2}{\epsilon_{(k)}^{*2} } - k^{1+\alpha}\cdot \epsilon_{(k)}^* \cdot \left(\sum_{i \in S_k}c_i/k \right) \right],
    \]
    where $\epsilon_{(k)}^* = \arg \min_{\epsilon \in \mathcal{R}_{\epsilon}(S_k)} \textsf{SC}_f(k, \epsilon)$. 
    \item Finally, the estimator variance achieved by the same coalition would be given by: 
    \[
         \textsf{Var}_f(\hat \mu ~|~k, \epsilon_{(k)}^*) = \frac{1}{k}\left(\sigma^2 + \frac{2}{\epsilon_{(k)}^{*2}} \right).
    \]
\end{itemize}
\end{lemma}
While the above lemma provides general insights for any given coalition size $k$, we are also interested in investigating some special cases of $k$ --- in particular, $k = n$ (the grand coalition) and $k = \Theta(1)$ (coalitions of small size in comparison to $n$). Our main objects of interest are: i) when is the set $\mathcal{R}_{\epsilon}(S_k)$ non-empty? ii) if $\mathcal{R}_{\epsilon}(S_k) \neq \emptyset$, how does the optimal $\epsilon^*$ depend on the size of the coalition $k$? The following result answers these questions: 

\begin{lemma}\label{lem:partial_grand}
When the number of players $n$ is sufficiently large and under partial decentralization,
\begin{itemize}
    \item For $\alpha < \frac{1}{2}$, there exists $\epsilon$'s for which the grand coalition $S_n$ is Nash-stable. In particular, $\epsilon_{(n)}^* = \left(\frac{4}{n^{1+\alpha} \bar c_n} \right)^{1/3}$, where $\bar c_n = \sum_{i=1}^n c_i/n$. 
    \item For $\alpha > \frac{1}{2}$, any Nash-stable coalition that exists has size $\Theta(1)$ and exists at $\epsilon$ of order $\Theta(1)$. 
\end{itemize}
\end{lemma}
While the above result does not argue about the optimal stable coalitions in different $\alpha$ regimes, it provides key building blocks for Theorem~\ref{thm:partial_cent}. For example, the first part shows that under partial decentralization, for $\alpha < \frac{1}{2}$, we can definitely find a Nash-stable coalition (the grand coalition) which is i) of the same size as the optimal coalition in the fully centralized setting; and ii) exists at a value of $\epsilon$ which exactly matches the order of the privacy levels chosen by the designer under full centralization (recall that, at the optimal coalition under full centralization, the designer chooses personalized privacy levels for each participant, but all these privacy levels have order $\Theta(n^{-(1+\alpha)/3})$). This provides clear intuition that there cannot possibly exist a large efficiency gap between the fully centralized and partially decentralized setting at least in this regime. Formally,  

\begin{thm}\label{thm:partial_cent}
When the number of players $n$ is sufficiently large, 
\begin{enumerate}
    \item for $\alpha \in \left[-1, \frac{1}{2}\right)$, the optimal Nash-stable coalition in the partially decentralized setting achieves variance and social cost of the orders $\Theta(n^{(2\alpha-1)/3})$ and $\Theta(n^{2(\alpha+1)/3})$ respectively; 
    \item for $\alpha \in \left[\frac{1}{2}, 1\right]$, the optimal Nash-stable coalition in the partially decentralized setting (if it exists) achieves variance and social cost of the orders $\Theta(1)$ and $\Theta(n)$ respectively. 
\end{enumerate}
\end{thm}
The theorem shows something even stronger. When $n$ is sufficiently large, there exists no gap between the fully centralized and partially decentralized data sharing mechanisms in terms of efficiency (both in social cost and estimator variance) beyond constant factors \textbf{across all privacy cost regimes}. I.e., the orders of social cost and estimator variance are exactly matched (for comparison, see Theorem~\ref{thm:kstar_c}). We will formalize this later through the Price of Stability in Section~\ref{sec:pos}.

\section{A Comparative Analysis across Different Centralization Regimes}\label{sec:pos}

We conclude by building on the results of Sections~\ref{sec:baseline},~\ref{sec:autonomous} and \ref{sec:partial} to explicitly compare how mechanisms with different degrees of decentralization fare in terms of efficiency of outcomes against the baseline mechanism where the central designer forms coalitions in a socially optimal way.

\subsection{Price of Stability for the Fully Decentralized Mechanism}

\paragraph{Social Cost Comparison.} Our first result in this section computes the price of stability of the fully decentralized mechanism with respect to \emph{social cost} across different $\alpha$ regimes. Note that these results apply to both the Nash and robust stability definitions.

\begin{thm}\label{thm:pos}
The Price of Stability for the fully decentralized mechanism with respect to social cost satisfies:
\begin{itemize}
    \item $\textsf{PoS}(\textsf{SC}) = \Theta\left(n^{2/3} \right)$ when $\alpha \in \left[-1, -\frac{1}{2}\right]$;
    \item $\textsf{PoS}(\textsf{SC}) = \Theta\left(n^{(1-2\alpha)/3} \right)$ when $\alpha \in \left(-\frac{1}{2}, \frac{1}{2}\right]$; and finally, 
    \item $\textsf{PoS}(\textsf{SC}) = \Theta(1)$ when $\alpha \in \left(\frac{1}{2}, 1\right]$. In particular, for this regime of $\alpha$, we have: 
     \[
        1 \leq \textsf{PoS}(\textsf{SC}) < \max \left[\frac{4}{3}, \left(\frac{2^{1/3}\sigma^2}{3c_{min}^{2/3}} \right) \right]. 
     \]
\end{itemize}
\end{thm}
\begin{proofsketch}
The order of the price of stability for the first two regimes follows directly from Theorems~\ref{thm:kstar_c} and \ref{thm:opt_eq}. However, these theorems just tell us that the price of stability for the regime $\left(\frac{1}{2}, 1 \right]$ is $\Theta(1)$ which is not that informative. Since $\textsf{PoS}(\textsf{SC}) \geq 1$ trivially, our goal is to derive a constant upper bound on the price of stability in the last regime. We do that by constructing lower bounds on the social cost of the optimal coalition in the centralized regime (using information about the curvature of the function $\textsf{SC}_c(k)$ from Equation~\eqref{eq:scbest_c}). For the social cost of the optimal stable coalition, we use the trivial upper bound of $n\sigma^2$. This is sufficient to obtain a constant upper bound of $\max \left[\frac{4}{3}, \left(\frac{2^{1/3}\sigma^2}{3c_{min}^{2/3}} \right) \right]$ on the $\textsf{PoS}$. 
\end{proofsketch}
There are a few key takeaways from the above theorem. First, there is a clear gap between the baseline and the fully decentralized mechanism for $\alpha < \frac{1}{2}$. The gap is intuitive for the intermediate regime of $\left(-\frac{1}{2}, \frac{1}{2}\right)$: there, the grand coalition is optimal for the centralized mechanism with sublinearly growing social costs; while the decentralized mechanism has social costs that grow linearly. The efficiency gap appears to be driven by the stability constraints of the decentralized setting. 

Interestingly, however, this gap persists and is the largest for $\alpha \leq - \frac{1}{2}$. This may be counter-intuitive because this is the regime where the grand coalition is optimal in the centralized setting and optimal (within constant factors) in the decentralized setting. This highlights that there are inefficiencies in the decentralized mechanism due to reasons beyond stability alone. The primary source of inefficiency is that when allowed autonomy, players opt for more stringent privacy levels (smaller $\epsilon$'s) compared to the centralized setting for the same coalition $S$. This increases accuracy costs for everyone while not saving appreciably on privacy costs, making the social cost worse. 

Finally, for $\alpha > \frac{1}{2}$, we highlight that for sufficiently large $n$, the upper bound on the $\textsf{PoS}$ is $\frac{4}{3}$. The other component of the upper bound becomes useful only when $n$ is small, in particular, if $n < \left( \frac{\sigma^2 }{(1+\alpha)(2c_{min})^{2/3}} \right)^{ 3/(2\alpha-1) }$.

\paragraph{Variance Comparison.} Similarly, we can compute the price of stability in terms of the \emph{estimator variance} (again for both our notions of stability):

\begin{thm}\label{thm:pos_var}
The Price of Stability for the fully decentralized mechanism with respect to estimator variance satisfies:
\begin{itemize}
    \item $\textsf{PoS}(\textsf{Var}) = \Theta\left(n^{2/3} \right)$ when $\alpha \in \left[-1, -\frac{1}{2}\right]$;
    \item $\textsf{PoS}(\textsf{Var}) = \Theta\left(n^{(1-2\alpha)/3} \right)$ when $\alpha \in \left(-\frac{1}{2}, \frac{1}{2}\right]$; and finally, 
    \item $\textsf{PoS}(\textsf{Var}) = \Theta(1)$ when $\alpha \in \left(\frac{1}{2}, 1\right]$. 
\end{itemize}
\end{thm}
It is interesting to note here that Theorems~\ref{thm:pos} and \ref{thm:pos_var} are identical in terms of the order of the price of stability. This is because the accuracy cost/variance term is either of the same order as the privacy cost term (as in the centralized mechanism), or strictly dominates the latter (as in the decentralized mechanism). Basically, the order of the variance
determines the order of the social cost in both centralized and decentralized mechanisms. 

Further, we can interpret the different $\alpha$ regimes as follows: i) When $\alpha \in \left[-1, -\frac{1}{2}\right)$, both centralized and decentralized mechanisms achieve non-trivial accuracy improvements, but there is a \textit{large efficiency gap} between centralized and decentralized. ii) When $\alpha \in \left[-\frac{1}{2}, \frac{1}{2}\right)$, only the centralized mechanism achieves non-trivial accuracy improvements, but the efficiency gap between centralized and decentralized diminishes with increasing $\alpha$. Finally, for $\alpha \in \left[\frac{1}{2}, 1\right]$, both centralized and decentralized mechanisms achieve trivial accuracy improvements over the outside option, but now they are within a constant factor of each other in terms of efficiency.  

\subsection{Price of Stability for the Partially Decentralized Mechanism}
Drawing from our earlier analysis, we developed the intuition that allowing players agency to choose their own privacy levels is the primary source of inefficiency in the fully decentralized mechanism. In this segment, we formally confirm that hypothesis: we show that opting for a partially decentralized mechanism where the players continue to make participation decisions, but the designer chooses privacy levels centrally, bridges the efficiency gap completely (down to constant factors). Formally, 

\begin{thm}\label{thm:pos_partial}
The Price of Stability for the partially decentralized mechanism with respect to both social cost and estimator variance is $\Theta(1)$ for all privacy cost regimes, i.e., $\textsf{PoS}(\textsf{SC}) = \Theta(1)$ and $\textsf{PoS}(\textsf{Var}) = \Theta(1)$ for all $\alpha \in [-1, 1]$.
\end{thm}
This proof follows directly from Theorems~\ref{thm:kstar_c} and~\ref{thm:partial_cent}. 
We can also quantify the constants when it comes to the price of stability for the social cost. In particular, for $\alpha < \frac{1}{2}$ and for large $n$, $\textsf{PoS}(\textsf{SC}) < \left(\frac{c_{max}}{c_{min}} \right)^{2/3}$, while for $\alpha > \frac{1}{2}$, we can continue to have the same upper bound on $\textsf{PoS}(\textsf{SC})$ as in Theorem~\ref{thm:pos}. 

In summary, we believe this is to be a powerful result. Because it shows that when it comes to mean-estimation problems, many of the data-sharing protocols that are prevalent in practice and use similar partially decentralized mechanisms, are actually very close to \textit{optimal} --- as long as the designer chooses the central privacy level appropriately.

\section{Conclusion and Future Work}\label{sec:conclude}

We studied different models of coalition formation for data sharing under differential privacy when agents have heterogeneous privacy preferences and privacy costs depend on the observation/attack model (and thus on the size of the coalition). Our settings ranged from full decentralization (where players choose both participation and privacy levels) to intermediate or partial decentralization (where players just make participation decisions and a privacy parameter is chosen centrally by the designer) to full centralization (where the designer has full control over coalition composition and privacy levels). There are many practical  benefits to data sharing in a decentralized manner which motivated this study: i) it gives agents more control over their data and ii) removes the need for trust and control through a central entity. Our goal was to investigate: i) whether full decentralization is efficient, and ii) if not, what level of decentralization is. 

We show that under full decentralization, the picture can often be grim: non-trivial accuracies can be achieved (relative to not sharing any data at all) only when privacy costs decrease sufficiently fast with coalition size ($\alpha < -\frac{1}{2}$); that is, only under regimes in which there is a non-negligible level of privacy amplification through data aggregation as more agents join. If privacy is measured at the local level, or worse, if privacy costs increase as more participants observe a player’s data, then it is not possible to obtain non-trivial accuracy in a fully decentralized fashion. This is in sharp contrast to the fully centralized baseline, where non-trivial improvements can be achieved even when privacy costs increase with coalition size, up to $\alpha = \frac{1}{2}$. Even for $\alpha < -\frac{1}{2}$, where full decentralization can be somewhat useful, it still has a large efficiency gap compared to the baseline mechanism which grows in the number of players $n$. This happens because selfish agents choose their privacy parameters $\epsilon_i$ inefficiently, opting for privacy levels that are significantly more stringent than what would be socially optimal, leading to lower privacy costs but inaccurate estimates. Interestingly, if we sacrifice some degree of autonomy for players, specifically allowing the central designer to pick a fixed $\epsilon$ for everyone (while players still make their own participation decisions), then we show that it is possible to recover all the nice properties of the baseline mechanism while reducing the efficiency gap down to constant factors.  

\paragraph{Future work.} We conclude by highlighting several directions for future work. First, our analysis focuses on moment estimation tasks which are a central primitive in modern machine learning. In particular, many widely used learning algorithms---such as stochastic gradient descent---rely on repeatedly averaging local updates or gradients across agents. A natural next step is to extend our analysis to capture repeated aggregation tasks, in higher dimensions. 

A second important direction is to understand what happens when agents have less information about each other's privacy preferences. This can lead to novel interactions and interesting game-theoretic models such as dynamically learning other agents' privacy preferences over repeated interactions and designing mechanisms that ensure truthful reporting of privacy parameters. 

Finally, our model adopts differential privacy as a formal notion of privacy loss. While DP provides strong theoretical guarantees and is natural for statistical computations and model training, practitioners often rely on large variety of privacy-preserving mechanisms when sharing data, including restricted data access, synthetic data generation, legal privacy protections, to only name a few. Extending our analysis to alternative privacy notions is an important open direction.

\section*{Acknowledgement}
RB acknowledges support from the US National Science Foundation (NSF) under grants CAREER-2144532 and 2112471. Part of this work was conducted when he was a visiting faculty researcher at Google. KD acknowledges support through the MIT METEOR Post-doctoral Fellowship. JZ and DS acknowledge support from the US NSF under grants IIS-2504990 and IIS-2336236. Part of this work was also supported by the Simons Institute for the Theory of Computing, and conducted when authors JZ and KD were visiting the Institute. 
Any opinions and findings expressed in this material are those of the authors and do not reflect the views of their funding agencies.

\bibliographystyle{plainnat}
\bibliography{arxiv2_bib.bib}

\appendix


\section{Additional Results with Proofs}\label{app:extra}

\subsection{Computation of Variance in Equation~\eqref{eq:variance}}
Recall that $x_i \sim \mathcal{D}$ denotes player $i$'s data point with $\textsf{Var}(x_i) = \sigma^2$ for all $i \in [n]$. Player $i$ adds noise $Z_i \sim Lap\left(1/\epsilon_i\right)$ to their data point and shares $\widetilde x_i = x_i + Z_i$ with the coalition. Note that $Z_i \perp x_i$. Therefore, for coalition $S \subseteq [n]$, the variance of the population mean estimator $\hat \mu$ is given by: 
\begin{align*}
    \textsf{Var}(\hat \mu(S, \vec \epsilon)) &= \textsf{Var} \left( \frac{\sum_{i \in S}\widetilde x_i}{|S|} \right) \\
    &= \frac{1}{|S|^2}\sum_{i \in S}\left( \textsf{Var}(x_i) + \textsf{Var}(Z_i)\right) \quad \text{(since $Cov(x_i, Z_i) = 0$, $Cov(\widetilde x_i, \widetilde x_j) = 0$)} \\
    &= \frac{1}{|S|^2}\sum_{i \in S} \left(\sigma^2 + \frac{2}{\epsilon_i^2} \right)\\
    &= \frac{\sigma^2}{|S|} + \frac{2}{|S|^2}\left(\sum_{i \in S}\frac{1}{\epsilon_i^2} \right). 
\end{align*} 

\subsection{Equilibrium Conditions for Nash Stability}

\begin{claim}\label{clm:eq_conditions_Nash}
Given ($\vec c, \sigma^2$) and any $\alpha \in [-1, 1]$, a subset $S \subseteq [n]$ with $|S| \triangleq k \geq 2$ is a Nash-stable coalition if and only if it satisfies the following conditions: 
\[
    \frac{(k - 1)}{k^{(2\alpha+1)/3} } \geq \frac{\sum_{j \in S} \left(c_j^2/2 \right)^{1/3} + 2\max_{j \in S}(c_j^2/2)^{1/3}}{\sigma^2}, \quad \text{and}
\]
\[
    k(k+1) < \left(\frac{3\min_{j \notin S} (c_j^2/2)^{1/3} }{\sigma^2}\right) (k+1)^{(2\alpha+4)/3} +  \left(\frac{\sum_{j \in S}(c_j^2/2)^{1/3} }{\sigma^2}\right) k^{(2\alpha+4)/3}.
\]
\end{claim}
\begin{proof}
Recall that given any coalition $S$, all burden-minimizing players in $S$ choose privacy levels according to $\vec \epsilon^*(S)$ (Lemma~\ref{lem:autonomous_eps}). Now, in order to derive the equilibrium conditions for $S$ to be Nash-stable, we will use the two conditions for Nash-stability. Firstly, $S$ must be stable to unilateral exit of players at $\vec \epsilon^*(S)$. This means that for any $i \in S$, 
\begin{align*}
       &  \frac{\sigma^2}{k} + \frac{2}{k^2}\left(\sum_{j \in S}\frac{1}{(\epsilon_j^*)^2}\right) + c_i(k) \epsilon_i^* \leq \sigma^2 \\
       \iff & \sigma^2 \left( \frac{k-1}{k} \right) \geq \frac{2}{k^2}\left(\sum_{j \in S}\frac{1}{(\epsilon_j^*)^2}\right) + c_i(k) \epsilon_i^* \\
       \iff & \sigma^2 \left(k-1 \right) \geq \frac{2}{k} \sum_{j \in S} \left( \frac{k^{2+\alpha} c_j }{4} \right)^{2/3} + k (k^{\alpha} c_i )\cdot \left( \frac{4}{k^{2+\alpha} c_i} \right)^{1/3} \quad \text{substituting $\epsilon_i^* = \left( \frac{4}{k^{(2+\alpha)}c_i} \right)^{1/3}$} \\
       \iff & \left(k-1 \right) \geq k^{(2\alpha+1)/3}   \left( \frac{\sum_{j \in S} \left(c_j^2/2 \right)^{1/3} + (4c_i^2)^{1/3}}{\sigma^2}  \right) \\
       \iff & \frac{(k - 1)}{k^{(2\alpha+1)/3}} \geq  \frac{\sum_{j \in S} (c_j^2/2)^{1/3} + 2(c_i^2/2)^{1/3}}{ \sigma^2}.
\end{align*}
Therefore, the unified condition for no player in $S$ to have incentives to leave, is given by: 
\begin{align*}
       \frac{(k-1)}{k^{(2\alpha+1)/3} } \geq \frac{\sum_{j \in S} \left(c_j^2/2 \right)^{1/3} + 2\max_{j \in S}(c_j^2/2)^{1/3}}{ \sigma^2}. 
\end{align*}
$S$ must also be stable with respect to outside players not having incentives to join unilaterally. This means that, for all $l \notin S$, we must have: 
\begin{align*}
    \max_{\epsilon_l \geq 0} \frac{\sigma^2}{(k+1)} + \frac{2}{(k+1)^2}\left( \frac{1}{\epsilon_l^2} + \sum_{j \in S}\frac{1}{(\epsilon_j^*)^2} \right) + c_l(k+1) \epsilon_l > \sigma^2
\end{align*}
The value of $\epsilon_l$ which maximizes the LHS is given by: 
\[
        \epsilon_l = \left( \frac{4}{(k+1)^2 c_l(k+1)} \right)^{1/3} = \left(\frac{4}{(k+1)^{(2+\alpha)}c_l} \right)^{1/3}.
\]
Plugging back, the unified condition is of the form: 
\begin{align*}
   k(k+1) < \left(\frac{3\min_{l \notin S} (c_l^2/2)^{1/3} }{\sigma^2}\right) (k+1)^{(2\alpha+4)/3} +  \left(\frac{\sum_{j \in S}(c_j^2/2)^{1/3} }{\sigma^2}\right) k^{(2\alpha+4)/3}.
\end{align*}
This concludes the proof of the claim. 
\end{proof}

\subsection{Equilibrium Conditions for Robust Stability}

\begin{claim}\label{clm:eq_conditions_robust}
Given $(\vec c, \sigma^2)$, a subset $S \subseteq [n]$ with $|S| \triangleq k \geq 2$ is a robust-stable coalition if and only if it satisfies the following conditions: 
\[
   \frac{(k - 1)}{k^{(2\alpha+1)/3} } \geq \frac{\sum_{j \in S} \left(c_j^2/2 \right)^{1/3} + 2\max_{j \in S}(c_j^2/2)^{1/3}}{\sigma^2}, \quad \text{and}
\]
\[
    \frac{k}{(k+1)^{(2\alpha+1)/3}} < \min_{l \notin S} \left[ \frac{ \sum_{j \in S\cup \{l\}} (c_j^2/2)^{1/3} + 2\max_{j \in S \cup \{l\}} (c_j^2/2)^{1/3} }{\sigma^2} \right].
\]
\end{claim}
\begin{proof}
The first condition follows from our earlier proof of Claim~\ref{clm:eq_conditions_Nash} and comes from the condition that no player already in $S$ has an incentive to leave $S$ unilaterally. We just need to verify that the second condition is true if and only if no player $l \notin S$ can make a valid entry into $S$. We proceed as follows.

When some player $l \notin S$ tries to enter into $S$, all players in $S \cup \{l\}$ have an incentive to choose their privacy levels in a way that maximizes their individual utilities. From our earlier calculation, we know that: 
\[
    \forall~j \in S \cup \{l\}, \quad \epsilon_j^*(S \cup \{l\}) = \left(\frac{4}{(k+1)^2 c_j(k+1)} \right)^{1/3} = \left(\frac{4}{(k+1)^{(2+\alpha)}c_j} \right)^{1/3},
\]
where $k = |S|$. 
However, player $l$'s valid entry will be blocked if and only if there exists some $j' \in S \cup \{l\}$ for whom staying in the coalition is not favorable, even after the re-adjustment of privacy levels as above, i.e., 
\[
    \exists~j' \in S \cup \{l\} \quad \text{s.t. } \frac{k}{(k+1)^{(2\alpha+1)/3}}  < \left[ \frac{ \sum_{j \in S\cup \{l\}} (c_j^2/2)^{1/3} + 2(c_{j'}^2/2)^{1/3} }{ \sigma^2} \right],
\]
or equivalently, 
\[
    \frac{k}{(k+1)^{(2\alpha+1)/3}}  < \left[ \frac{ \sum_{j \in S\cup \{l\}} (c_j^2/2)^{1/3} + \max_{j \in S \cup  \{l\} } 2(c_j^2/2)^{1/3} }{\sigma^2} \right].
\]
Since the above must hold for all players $l \notin S$, we have the desired condition:
\[
    \frac{k}{(k+1)^{(2\alpha+1)/3}}  < \min_{l \notin S}\left[ \frac{ \sum_{j \in S\cup \{l\}} (c_j^2/2)^{1/3} + \max_{j \in S \cup  \{l\} } 2(c_j^2/2)^{1/3} }{\sigma^2} \right].
\]
This concludes the proof of the claim. 
\end{proof}

\subsection{Existence and Properties of Stable Coalitions under Full Decentralization}

\paragraph{Equilibrium multiplicity and non-monotonicity in player cost.}
Here, we present results illustrating some properties of stable coalitions in the fully decentralized setting under both definitions of stability. 

\begin{lemma}\label{lem:multiplicity}
There exist problem instances $(\alpha, \vec c, \sigma^2)$ where multiple (Nash/robust) equilibrium coalitions exist simultaneously.
\end{lemma}

\begin{lemma}\label{lem:non_monotone_cost}
There exist problem instances $(\alpha,\vec c, \sigma^2)$ which admit a (Nash/robust) equilibrium coalition $S$ where participation is not-monotonic in player costs. I.e., there exist $i, j \in [n]$ such that $c_i \leq c_j$ and player $i \notin S$, but player $j \in S$. 
\end{lemma}

In order to prove Lemmas~\ref{lem:multiplicity} and~\ref{lem:non_monotone_cost}, it is sufficient to provide one problem instance ($\alpha, \vec c, \sigma^2$) where we can find more than one (Nash/robust) equilibrium which, by default, would imply that at least one of them exhibits the non-monotonicity property. 
\begin{proof}
Consider the following problem instance with $\alpha = 1$, $\sigma^2 = 0.25$ and the cost profile given by
    $\vec c = [1.80 \times 10^{-3}, 2.15 \times 10^{-3}, 2.20 \times 10^{-3}, 15 \times 10^{-3}, 15.5 \times 10^{-3}, 17 \times 10^{-3}]$.
We will quickly verify that coalitions $U_1 = \{1,2,3,4\}$ and $U_2 = \{1,2,3,5\}$ are indeed both Nash equilibrium coalitions for this problem instance (therefore, also robust equilibria, by Claim~\ref{clm:Nash_stronger}). 
\begin{align*} 
    \vec \epsilon(U_1) &= [3.2624, 3.0748, 3.0513, 1.6091]; \quad \textsf{Var}(\hat \mu~|~U_1) = 0.1492; \\ 
    \quad B(U_1) &= [ 0.1727, 0.1756, 0.1760, 0.2457]; \quad B(U_1^c) = [0.2535, 0.2629].
\end{align*}
\begin{align*}
     \vec \epsilon(U_2) &= [3.2624, 3.0748, 3.0513, 1.5917]; \quad \textsf{Var}(\hat \mu~|~U_2) = 0.1502; \\ 
     \quad B(U_2) &= [0.1737, 0.1767, 0.1771, 0.2489]; \quad B(U_2^c) = [0.2510, 0.2636].\\
\end{align*}
For both $U_1$ and $U_2$, we compute the player-wise optimal $\epsilon$'s, the estimator variance, the participation burdens of players in the coalition (given by $B(U_1)$ and $B(U_2)$) and the participation burdens of players outside the coalition if they were to join unilaterally (given by $B(U_1^c)$ and $B(U_2^c)$). Observe that all participation burdens of players in coalitions are better than the outside option $\sigma^2 = 0.25$ which means that none of them have an incentive to leave, while external players who try to join unilaterally would always be worse than $\sigma^2$. 
\end{proof}

The intuition here is that the costs of players $4$ and $5$ are so close-by that they are effectively \textit{interchangeable} without affecting players $1$, $2$ and $3$. However, note that players $4$ and $5$ cannot simultaneously be part of an equilibrium coalition along with the first $3$ players---this would necessitate players $1$, $2$ and $3$ to choose more stringent privacy levels (because the coalition size is larger) which would make joining infeasible for player $5$. This leads to the existence of an equilibrium coalition ($U_2$) which is non-monotonic in player costs. At a high level, Lemmas \ref{lem:multiplicity} and \ref{lem:non_monotone_cost} demonstrate that stable coalitions often lack `nice structure' for general cost profiles, which makes the problem of finding stable coalitions (under both Nash and robust definitions) challenging.  

\begin{remark}
Intuitively, it does appear (based on the problem instance we provided before) that if the costs of consecutive players are not so close together, some of the problems articulated above (like equilibria being non-monotonic in player costs and equilibria multiplicity) may be alleviated to some extent. That is, indeed, true. We can show that if the cost profile $\vec c$ is ``well-separated", any Nash equilibrium that exists must be downward closed. In addition to solving the non-monotonicity issue, it gives us a better handle on equilibria multiplicity. Instead of having to search through all possible subsets of $[n]$ of size $\geq 2$, it suffices to check only the $n-1$ coalitions of the form $S_2, S_3,...S_n$. The details of the result (Claim~\ref{clm:separated_cost_Nash}) can be found in Appendix~\ref{app:extra}. 
\end{remark}

\paragraph{Equilibrium Existence.} 
We now highlight a structural property related to equilibrium existence that differ across our two notions of stability. This pertains to whether the existence of equilibrium coalitions is monotonic in the variance $\sigma^2$. Lemma \ref{lem:non_monotone_var} shows that existence of Nash-stable coalitions is non-monotonic in $\sigma^2$, there are settings where a Nash-stable coalition might exist at a lower variance level, but no coalition may exist at a higher variance level. On the other hand, such non-monotonicity is \emph{not} true of robust stability: as Lemma \ref{lem:robust_weak_monotonicity} shows, both the existence of a robust equilibrium and the maximum size of a robust equilibrium at any variance level $\sigma^2$ must always be monotonic in $\sigma^2$.

\begin{lemma}\label{lem:non_monotone_var}
The existence of a Nash equilibrium coalition may be non-monotonic in the variance level $\sigma^2$. That is, if a Nash equilibrium exists at $\sigma^2$, there may exist $\sigma' > \sigma$ where no Nash equilibrium exists. 
\end{lemma}

\begin{proof}
In order to prove this result, we need to provide a problem instance ($\alpha, \vec c, \sigma^2$) and some $\sigma'$ ($\sigma' > \sigma$) for which at least one Nash equilibrium exists at $\sigma^2$, but none exist at $\sigma'^2$. 
Consider the following problem instance: $\alpha = 1$, cost profile is given by $\vec c
= \big[2.2\times 10^{-4},\ 5.4\times 10^{-4},\ 7.0\times 10^{-4},\ 11\times 10^{-4},\ 
30\times 10^{-4},\ 33\times 10^{-4},\ 34\times 10^{-4},\ 36\times 10^{-4},\ 38\times 10^{-4}\big]$ and $\sigma^2 = (0.25)^2$. Further, choose $\sigma' = 0.4$. As we see in Figure~\ref{fig:eq_var}, a Nash-stable equilibrium exists at $\sigma^2$, but not at $\sigma'^2$. 
\end{proof}

\begin{figure}[!ht]
    \centering
    \subfloat[Nash equilibria]{\includegraphics[width=0.45\textwidth]{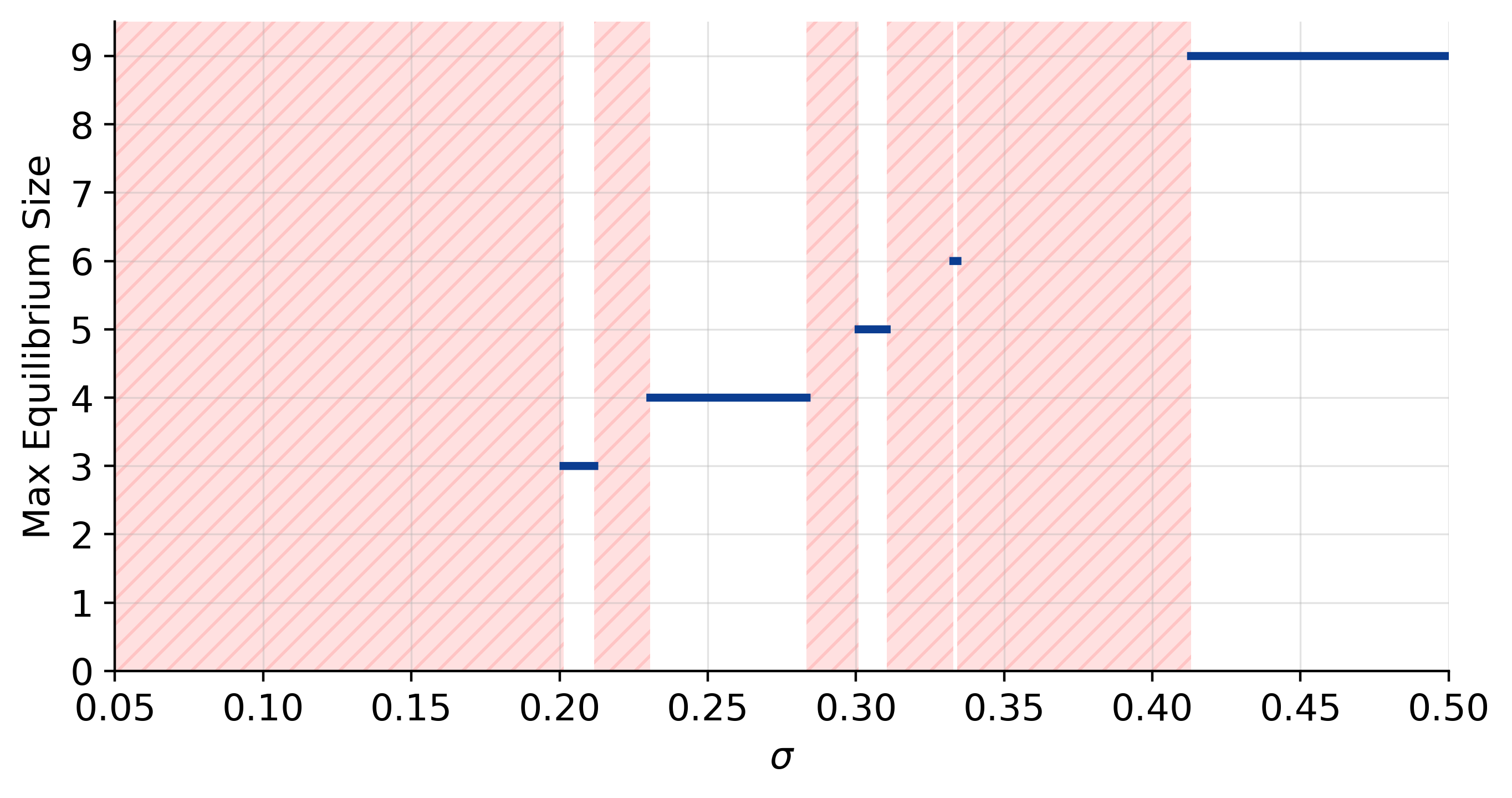}} 
    \subfloat[Robust Equilibria]{\includegraphics[width=0.45\textwidth]{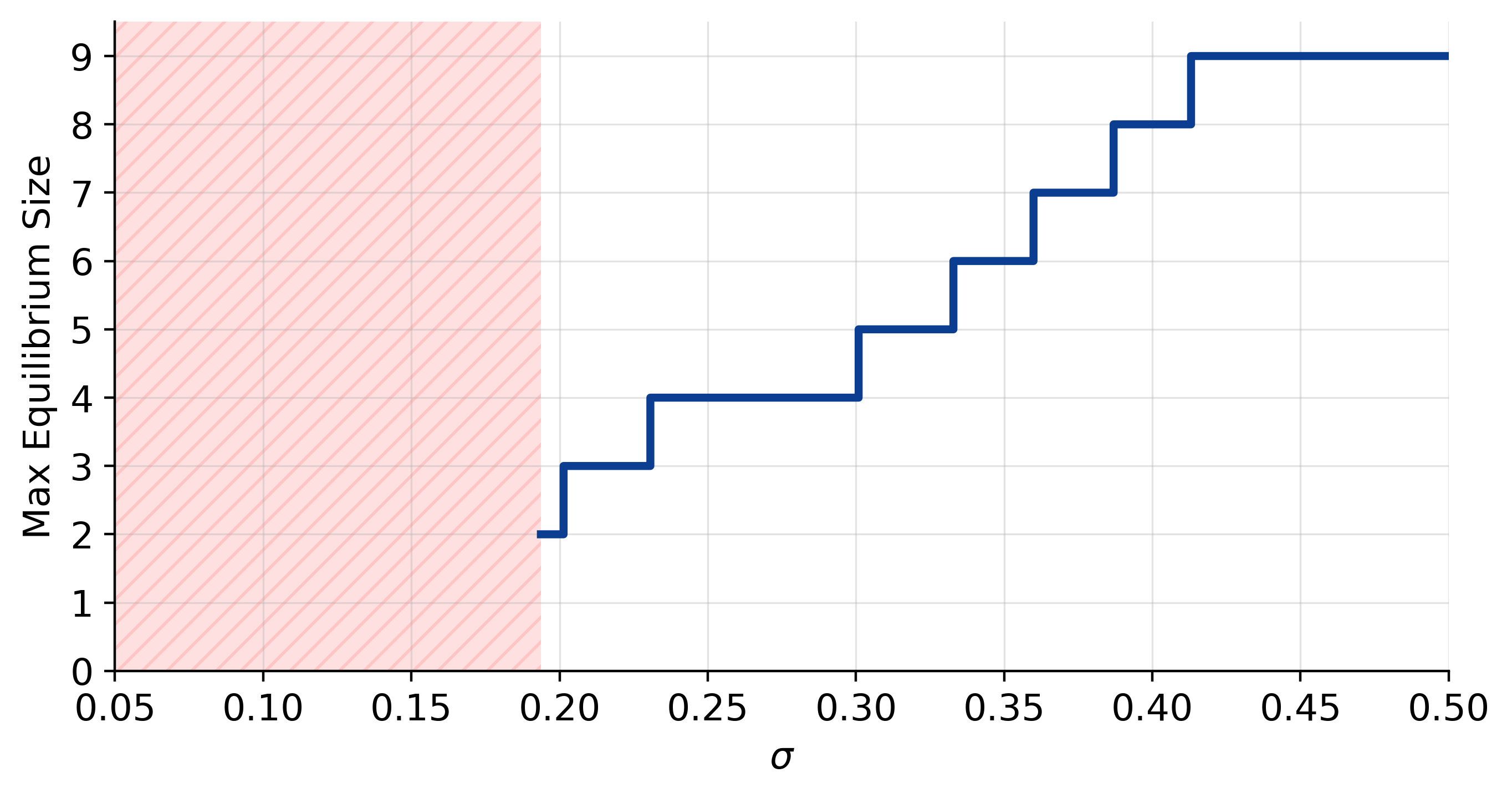}} 
    \caption{We plot the maximum size of equilibrium coalition that exists under the Nash definition (left) and the robust definition (right) as a function of $\sigma$. Parameters of the problem instance: cost profile $\vec c = \big[2.2\times 10^{-4},\ 5.4\times 10^{-4},\ 7.0\times 10^{-4},\ 11\times 10^{-4},\ 
    30\times 10^{-4},\ 33\times 10^{-4},\ 34\times 10^{-4},\ 36\times 10^{-4},\ 38\times 10^{-4}\big]$ and $\alpha = 1$. The shaded regions indicate $\sigma$ values for which no equilibrium exists under that stability definition. While the Nash definition exhibits non-monotonicity even in equilibrium existence, the robust definition exhibits monotonicity, not only in existence, but also in maximum equilibrium size.}
    \label{fig:eq_var}
\end{figure}
To understand this result, recall from Lemma \ref{lem:autonomous_eps} that a player's choice of $\epsilon$ is influenced solely by their privacy cost $c_i$ and the size of the coalition. Thus, for a coalition $S$ existing at variance level $\sigma^2$, as $\sigma^2$ increases, every player in $S$ will still keep their $\epsilon$ levels unchanged and continue to be a part of the coalition. However, with the increasing variance, it is possible that some player $j$ outside of the coalition $S$ may now find participation the better option. This immediately breaks Nash stability and coalition $S$ dissolves. This explains the non-monotonicity of existence in $\sigma^2$ for Nash-stable coalitions. However, when moving to the robust equilibrium concept, these non-monotonicities in $\sigma^2$ disappear.

\begin{lemma}\label{lem:robust_weak_monotonicity}
Given a cost profile $\vec c$, suppose that there exists a robust-stable coalition of size $k$ at variance level $\sigma^2$. Then for all $\sigma' \geq \sigma$, there must exist at least one robust-stable coalition of size $\geq k$ at variance level $\sigma'^2$. 
\end{lemma}
\begin{proof}
Let $S \subseteq [n]$ be a robust equilibrium coalition of size $k$ at variance level $\bar \sigma^2$. First, we define $S_0 = S$. Then we note that we can construct a unique finite sequence of sets $\{S_i\}_{i \geq 1}$ with $|S_i| = |S_{i-1}|+1$ and $S_i = S_{i-1} \cup \{k'(i)\}$, where: 
\[
k'(i) = \arg\min_{l \notin S_{i-1}} \left[ \sum_{j \in S_{i-1}\cup \{l\}} (c_j^2/2)^{1/3} + \max_{j \in S_{i-1} \cup \{l\}} 2(c_j^2/2)^{1/3}  \right] \quad \forall~i \in [n-k]. 
\]
For $n-k \geq i \geq 0$, define: 
\begin{align*}
      T_i := \left(\sum_{j \in S_i}(c_j^2/2)^{1/3} + 2\max_{j \in S_i}(c_j^2/2)^{1/3} \right) \frac{|S_i|^{(2\alpha+1)/3}}{(|S_{i}|-1)},
\end{align*}
with $T_{n-k+1} = \infty$. According to definition of robust stability, for $i\geq 1$, $S_i$ would be a robust equilibrium coalition for all $\sigma \in \left[\sqrt{T_i}, \sqrt{T_{i+1}} \right)$, provided the interval is feasible (otherwise, $S_i$ would not be a robust equilibrium coalition). In order to complete our proof, we need to argue that for any $\sigma > \bar \sigma$, there must exist some $m$ with $n-k \geq m \geq 0$ for which the interval $\left[ \sqrt{T_m}, \sqrt{T_{m+1}} \right)$ is feasible and $\sigma$ belongs to that interval. This follows directly because the end-point and the start-point of consecutive intervals coincide and $\left[\sqrt{T_{n-k}}, \sqrt{T_{n-k+1}} \right)$ is always feasible (and of infinite width). Further, $|S_m| \geq |S_0| = k$ by construction of the sequence $\{S_i\}$ since $m \geq 0$. This concludes the proof of the lemma. 
\end{proof}

The intuition here is that as we increase variance from $\sigma^2$, even after the original coalition $S$ dissolves under the Nash definition, $S$ will still continue to exist (with the originally chosen $\epsilon$ levels as $\sigma^2$) under the robust definition --- until the new player $j$ can make a \textit{valid entry into $S$}. Let $\sigma_{tr}^2$ be the transition value of the variance at which player $j$ can finally make valid entry. If the variance is increased all the way up to $\sigma_{tr}^2$, $S$ will finally cease to be robust-stable, but the augmented coalition $S \cup \{j\}$ will immediately be robust-stable (with all players re-adjusting their $\epsilon$ levels to $\vec \epsilon^*(S \cup \{j\})$). This illustrates, at an intuitive level, why robust stability has monotonicity of existence and size in variance.

\subsection{Nash-stability under ``Well-Separated" Costs}

\begin{claim}\label{clm:separated_cost_Nash}
Consider a setting where player costs are ``well-separated", i.e., $c_i \geq 2 c_{i-1}$ for all $i \geq 2$. Then, 
\begin{itemize}
    \item any Nash-stable coalition that exists must also be downward-closed, 
    \item all Nash-stable coalitions that exist at any given variance level $\sigma^2$ can be found in $\mathcal{O}(n)$ time where $n$ is the number of players.  
\end{itemize}
\end{claim}

\begin{proof}
We will prove the first part by contradiction. Suppose, given a well-separated cost profile $\vec{c}$ and variance $\sigma^2$, there exists a Nash stable coalition $S$ of size $k$. Let player $l$ be the player with the largest cost $c_l$ in $S$. Further, for the sake of contradiction, we assume that there exists some players from $1$ through $l$ who are not included in $S$. Let $m$ be the smallest index of players excluded from $S$. Clearly, $m < l$. 

Now, since $S$ is Nash-stable, no player, including player $l$, has an incentive to leave $S$. This means that: 
\[
 \frac{(k- 1)}{k^{(2\alpha+1)/3}} \geq \left( \frac{ \sum_{j \in S}(c_j^2/2)^{1/3} + 2(c_l^2/2)^{1/3} }{\sigma^2}  \right).
\]
Similarly, no player currently excluded from $S$, in particular player $m$, has an incentive to join the coalition. This further implies that: 
\begin{align*}
  k(k+1) &< \left(\frac{3(c_m^2/2)^{1/3}}{\sigma^2} \right)(k+1)^{(2\alpha+4)/3} + \left( \frac{\sum_{j \in S}(c_j^2/2)^{1/3}}{\sigma^2} \right)k^{(2\alpha+4)/3} \\
  &< \left( \frac{3(c_m^2/2)^{1/3} + \sum_{j \in S}(c_j^2/2)^{1/3} }{\sigma^2} \right) (k+1)^{(2\alpha+4)/3} \\
  &< \left( \frac{2(c_l^2/2)^{1/3} + \sum_{j \in S}(c_j^2/2)^{1/3} }{\sigma^2} \right) (k+1)^{(2\alpha+4)/3}. \quad (\text{using the well-separated property of $\mathbf{c}$}) \\
\end{align*}
Combining both conditions, we have: 
\[
      \left( \frac{2(c_l^2/2)^{1/3} + \sum_{j \in S}(c_j^2/2)^{1/3} }{\sigma^2} \right)  \leq \frac{(k-1)}{k^{(2\alpha+1)/3}} < \frac{k}{(k+1)^{(2\alpha+1)/3}} < \left( \frac{2(c_l^2/2)^{1/3} + \sum_{j \in S}(c_j^2/2)^{1/3} }{\sigma^2} \right),
\]
which is clearly a contradiction. The key step above is noting that the function $f(x) = \frac{(x-1)}{x^{(2\alpha+1)/3}}$ is increasing in $x$ for all $x \geq 1$ and for all $\alpha \in [-1, 1]$. This concludes the proof that any Nash-stable coalition must be downward closed. \\

For any given size, there exists a unique coalition which is downward-closed. This implies that in order to find all possible Nash-stable coalitions at given variance level $\sigma^2$, we need to check only through sets of the form $S_k$ for $k \in \{2,3,....n\}$ which can be done in linear time in the number of players.
\end{proof}

\subsection{More Structural Properties of Robust Stability}

\begin{claim}\label{clm:eq_largest}
Given a cost profile $\vec c$ and $\alpha \in [-1, 1]$, let $k_{eq,max}$ be the size of the largest robust-stable coalition at variance level $\sigma^2$ (assuming it exists). If $\mathcal{S}_{max}$ represents the set of all robust-stable coalitions of size $k_{eq,max}$, then one of the coalitions in $\mathcal{S}_{max}$ must be downward-closed. 
\end{claim}
\begin{proof}
Let $S$ be a robust stable coalition of size $k_{eq,max}$ (there exists at least one by assumption, there may be more). 
If $S = S_{k_{eq,max}}$, we are done trivially. So, let $S \neq S_{k_{eq,max}}$, i.e., $S$ is not downward closed. Then, we will argue that there must exist at least one more robust equilibrium coalition at $\sigma^2$ which is of the form of $S_{k_{eq,max}}$. For any $U \subseteq [n]$, we define:
\[
      f(U) = \sum_{j \in U}(c_j^2/2)^{1/3} + \max_{j \in U}2(c_j^2/2)^{1/3}.
\]
For the sake of contradiction, suppose that $S_{k_{eq,max}}$ is not a robust equilibrium coalition at $\sigma^2$. Therefore, $S_{k_{eq,max}}$ must violate at least one of the equilibrium conditions. Now, since $S$ is an equilibrium coalition at $\sigma^2$, it must be that: 
\[
    \sigma^2\frac{(k_{eq,max}-1)}{k_{eq,max}^{(2\alpha+1)/3}} \geq f(S).  
\]
We can also directly observe that the $f(S) \geq f(S_{k_{eq,max}})$. This means that: 
\[
        \sigma^2\frac{(k_{eq,max}-1)}{k_{eq,max}^{(2\alpha+1)/3}} \geq f(S_{k_{eq,max}}),
\]
i.e., $S_{k_{eq,max}}$ already satisfies the first equilibrium condition (no player in $S_{k_{eq,max}}$ has incentive to leave). But this implies that
$S_{k_{eq,max}}$ must violate the other equilibrium condition, i.e., player $k_{eq,max}+1$ can join without the dissolution of the coalition, 
\[
       f(S_{k_{eq,max}+1}) \leq \sigma^2 \frac{k_{eq,max}}{(k_{eq,max}+1)^{(2\alpha+1)/3}}. 
\]
Note that the above would be one of the conditions required for the augmented coalition $S_{k_{eq,max}+1}$ to be robust-stable (which is already satisfied now). If the other condition is also satisfied, we get a robust equilibrium coalition of size $k_{eq,max}+1 > k_{eq,max}$. However, if the other condition is not satisfied, this means that the next player can also be added. We continue reasoning like this, either obtaining an equilibrium at some point and stopping, or adding the next player. 

Since the number of players is finite, this process eventually has to terminate (because the grand coalition has a one-sided equilibrium condition). This implies that when we stop, we will have a robust-stable coalition of size strictly greater than $k_{eq,max}$ which is additionally downward closed. But this would contradict our initial assumption that $k_{eq,max}$ is the size of the largest robust-stable coalition at level $\sigma^2$. This concludes the proof.  
\end{proof}

The above result tells us that the largest robust-stable coalition has some additional structure. This means that given $(\vec c, \sigma^2)$, at least the size of the largest robust-stable coalition can be found in linear time. For smaller values of variance, this can significantly prune the search space of all possible robust-stable coalitions and enable faster search.


\section{Proofs of Results in Section~\ref{sec:model}}\label{app:model}
\subsection{Proof of Claim~\ref{clm:Nash_stronger} (Nash-stability $\implies$ Robust-stability)}
Both our notions of stability have an identical first condition (no player in $S$ wants to leave unilaterally). Therefore, 
in order to complete the proof, we need to show that any set $S \subseteq [n]$ with $|S| \triangleq k \geq 2$ which satisfies: 
\[
    k(k+1) < \left(\frac{\min_{l \notin S} 3(c_l^2/2)^{1/3}  }{\sigma^2}\right) (k+1)^{(2\alpha+4)/3} +  \left(\frac{\sum_{j \in S}(c_j^2/2)^{1/3} }{\sigma^2}\right) k^{(2\alpha+4)/3},
\]
must also satisfy:
\[
   .\frac{k}{(k+1)^{(2\alpha+1)/3}}  < \min_{l \notin S} \left[ \frac{ \sum_{j \in S\cup \{l\}} (c_j^2/2)^{1/3} + 2\max_{j \in S \cup \{l\}} (c_j^2/2)^{1/3} }{\sigma^2} \right].
\]

We have: 
\begin{align*}
    &k(k+1) < \left(\frac{\min_{l \notin S} 3(c_l^2/2)^{1/3}  }{\sigma^2}\right) (k+1)^{(2\alpha+4)/3} +  \left(\frac{\sum_{j \in S}(c_j^2/2)^{1/3} }{\sigma^2}\right) k^{(2\alpha+4)/3} \\
    \implies & k(k+1) < \left(\frac{\min_{l \notin S} 3(c_l^2/2)^{1/3}  }{\sigma^2}\right) (k+1)^{(2\alpha+4)/3} +  \left(\frac{\sum_{j \in S}(c_j^2/2)^{1/3} }{\sigma^2}\right) (k+1)^{(2\alpha+4)/3} \quad \text{(since $\alpha >-2$)}\\
    \implies &k(k+1) < \left( \frac{ \sum_{j \in S} (c_j^2/2)^{1/3} + \min_{l \notin S} 3(c_l^2/2)^{1/3} }{\sigma^2} \right) (k+1)^{(2\alpha+4)/3} \\
    \implies &\frac{k}{(k+1)^{(2\alpha+1)/3}} < \min_{l \notin S} \left( \frac{ \sum_{j \in S \cup \{l\} }(c_j^2/2)^{1/3} + 2(c_l^2/2)^{1/3} }{\sigma^2} \right) \\
    \implies &\frac{k}{(k+1)^{(2\alpha+1)/3}} < \min_{l \notin S} \left[ \frac{ \sum_{j \in S\cup \{l\}} (c_j^2/2)^{1/3} + 2\max_{j \in S \cup \{l\}} (c_j^2/2)^{1/3} }{\sigma^2} \right].
\end{align*}


\section{Proofs of Results in Section~\ref{sec:baseline}}\label{app:centralized}

\subsection{Proof of Lemma~\ref{lem:centralized}}
Let $S \subseteq [n]$ be a coalition of size $k \geq 2$ picked the designer at $(\vec c, \sigma^2)$. Let $\vec \epsilon \in \mathbb{R}_{> 0}^k$ be the privacy levels chosen by the designer for players in $S$. Then the social cost at $(S, \vec \epsilon)$ would be given by: 
\begin{align*}
    \textsf{SC}(S, \vec \epsilon) &= (n-k)\sigma^2 + \sum_{i \in S} \left[ \frac{\sigma^2}{k} + \frac{2}{k^2}\left(\sum_{j \in S}\frac{1}{\epsilon_j^2}\right) + c_i(k)\epsilon_i \right] \\
    &= (n-k)\sigma^2 + \sigma^2 + \frac{2}{k}\left(\sum_{j \in S}\frac{1}{\epsilon_j^2}\right) + k^{\alpha}\left(\sum_{i \in S} c_i \epsilon_i\right) \\
    &= (n+1-k)\sigma^2 + \sum_{i \in S}\left( \frac{2}{k \epsilon_i^2} + k^{\alpha}c_i \epsilon_i \right).  
\end{align*}
Note that the social cost function is separable in $\epsilon_i$'s, so the designer can optimize each of them individually to achieve the best social cost at coalition $S$. The designer's optimal choice of $\epsilon_i$ for each player $i \in S$ can be uniquely obtained by solving the first order condition of the convex function $v(\epsilon_i) = \frac{2}{k \epsilon_i^2} + k^{\alpha}c_i \epsilon_i$: 
\[
           -\frac{4}{k \epsilon_i^3} + k^{\alpha}c_i = 0 \iff \epsilon_i = \left( \frac{4}{k^{(1+\alpha)}c_i} \right)^{1/3}. 
\]
Therefore, the best social cost achievable at coalition $S$ (with $|S| = k$) is given by:
\begin{align*}
    &(n+1-k)\sigma^2 + \sum_{i \in S}\left( \frac{2}{k \epsilon_i^2} + k^{\alpha}c_i \epsilon_i \right) \\
    = &(n+1)\sigma^2 - k\sigma^2 + \frac{1}{k} \sum_{i \in S}\left(\frac{2}{\epsilon_i^2} + k^{(1+\alpha)}c_i \epsilon_i \right) \\
    = &(n+1)\sigma^2 - k\sigma^2 + \frac{1}{k} \sum_{i \in S} \left( 2 \cdot \frac{ k^{2(1+\alpha)/3}c_i^{2/3} }{2^{4/3}} + k^{(1+\alpha)}c_i \cdot \frac{2^{2/3}}{ k^{(1+\alpha)/3}c_i^{1/3} } \right) \\
    = &(n+1)\sigma^2 - k\sigma^2 + \frac{1}{k} \sum_{i \in S} \left( \frac{ k^{2(1+\alpha)/3}c_i^{2/3} }{ 2^{1/3}} + 2^{2/3}\cdot  k^{2(1+\alpha)/3}c_i^{2/3}   \right) \\
    = & (n+1)\sigma^2 - k\sigma^2 + \left(\frac{3}{2^{1/3}}\right) k^{2(1+\alpha)/3}\left(\sum_{i \in S}c_i^{2/3}/k \right) \\
    = & (n+1)\sigma^2 - \left[ k\sigma^2 - \left(\frac{3}{2^{1/3}}\right) k^{2(1+\alpha)/3}\left(\sum_{i \in S}c_i^{2/3}/k \right) \right].
\end{align*}
Further, note that out of all coalitions of size $k$, the social cost would be minimized at the coalition which minimizes $\sum_{i \in S}c_i^{2/3}$ which clearly must be the coalition consisting of the $k$ smallest cost players, i.e., $S_k$ (which is downward closed). Therefore, putting everything together, the social cost at the optimal coalition of size $k$ would be given by: 
\[
         \textsf{SC}_c(k) = (n+1)\sigma^2 - \left[ k\sigma^2 - \left(\frac{3}{2^{1/3}}\right) k^{2(1+\alpha)/3}\left(\sum_{i \in S_k}c_i^{2/3}/k \right) \right].
\]
Note that although the social cost function $\textsf{SC}(\cdot)$ usually takes a coalition and an $\epsilon$-level vector as its inputs, the function $\textsf{SC}_c(\cdot)$ represents the social cost at the optimal coalition of size $k$, which is uniquely defined just by $k$. We can similarly compute the expression for the estimator variance at the optimal coalition of size $k$ as follows:
\begin{align*}
    \textsf{Var}_c(\hat \mu~|~k) = \frac{1}{k}\left[\sigma^2 + \frac{ k^{2(1+\alpha)/3}}{ 2^{1/3}} \left(\sum_{i \in S_k}c_i^{2/3}/k \right)  \right].
\end{align*}
This concludes the proof of the lemma. 

\subsection{Proof of Theorem~\ref{thm:kstar_c}}
Given ($\vec c, \sigma^2$), the social cost at the optimal coalition of size $k$ for $n \geq k \geq 2$ under full centralization is given by:
\[
          \textsf{SC}_c(k) = (n+1)\sigma^2 - \underbrace{\left[k\sigma^2 - \left(\frac{3}{2^{1/3}}\right) k^{2(\alpha+1)/3}\left(\sum_{i \in S_k}c_i^{2/3}/k \right)  \right]}_{g(k)}. 
\]
Further, for $k = 0$ (the empty coalition), the social cost takes a value of $n\sigma^2$. In order to find $k^*$ as defined in Equation~\eqref{opt:kstar}, we thus need to find where the function $g(k)$ is maximized for different regimes of $\alpha$:

\begin{enumerate}
    \item $\alpha \in \left[-1, \frac{1}{2}\right)$: First, we note that since all $c_i$'s come from the bounded interval $[c_{min}, c_{max}]$ with $c_{min} > 0$, we must have that $c_{min}^{2/3} \leq \sum_{i \in S_k}c_i^{2/3}/k \leq c_{max}^{2/3}$ for all $k \in \{2,3...n\}$. This also implies that for $k \geq 2$, $k \in \mathbb{Z}$, $g(k) \geq g_1(k)$ where $g_1(k)$ is given by: 
    \[
        g_1(k) = \left[k\sigma^2 - \left(\frac{3}{2^{1/3}}\right) k^{2(\alpha+1)/3}c_{max}^{2/3}  \right].
    \]
    Further, both $g(k)$ and $g_1(k)$ have the same leading term. Therefore, if we can show that $g_1(k)$ is monotonically increasing in $k$ once $k$ is large enough, $g(k)$ must also be monotonically increasing in $k$ for large $k$. The eventual monotonicity of $g_1(k)$ is easy to verify: the leading term $k\sigma^2$ grows linearly, while the second term $\left(\frac{3}{2^{1/3}}\right) k^{2(\alpha+1)/3}c_{max}^{2/3}$ only grows sub-linearly for $\alpha < \frac{1}{2}$. Therefore, for large $k$, the first term dominates and the overall function must increase monotonically. 
    So when $k \in \{2,3,...n\}$, choosing $k = n$ offers the best social cost. Plugging in, we obtain:
    \[
             (n+1)\sigma^2 - g(n) = \sigma^2 + \left(\frac{3}{2^{1/3}}\right) n^{2(\alpha+1)/3}\left(\sum_{i \in S_n}c_i^{2/3}/n \right),
    \]
    which is also less than $n\sigma^2$ (so, better than the empty coalition) for sufficiently large $n$. This completes the proof that $k^* = n$ for $\alpha < \frac{1}{2}$ which leads to $\textsf{SC}_c(k^*) = \Theta(n^{2(\alpha+1)/3})$ and $\textsf{Var}_c(\hat \mu~|~k^*) = \Theta(n^{(2\alpha-1)/3})$.
    \item $\alpha = \frac{1}{2}$: For $\alpha = \frac{1}{2}$, we have $g(k) = k \left[\sigma^2 - \left(\frac{3}{2^{1/3}}\right)\left(\sum_{i \in S_k}c_i^{2/3}/k \right) \right]$. Where $g(k)$ is maximized depends on the relative order of $\sigma^2$ and $\sum_{i \in S_k}c_i^{2/3}/k$ which determines the overall sign of the term. But, in all cases, due to $\sum_{i \in S_k}c_i^{2/3}/k$ being $\Theta(1)$, the social cost function is effectively linear in $k^*$. Further, $k^* = 0$ and $k^* = n$ both achieve social cost of order $\Theta(n)$. So any intermediate value of $k^*$ would also achieve $\textsf{SC}_c(k^*) = \Theta(n)$. For the same reason, we have $\textsf{Var}_c(\hat \mu~|~k^*) = \Theta(1)$. 
    \item $\alpha \in \left(\frac{1}{2}, 1\right]$: Finally, when $\alpha > \frac{1}{2}$, the leading term in $g(k)$ grows linearly in $k$ while the second term grows super-linearly (noting again that $\sum_{i \in S_k}c_i^{2/3}/k = \Theta(1)$). Therefore, overall, $g(k)$ either decreases monotonically in $k$ or it is concave in $k$ (so it will achieve its maxima somewhere before decreasing monotonically). In the first case when $g(k)$ is monotonically decreasing, $k^* = 0$. However when $g(k)$ is concave, $k^* = \Theta(1)$ and this follows from a similar argument we made in the $\alpha < \frac{1}{2}$ regime  --- $g(k)$ can be squeezed between two concave functions $g_1(k)$ and $g_2(k)$ both of which attain their maxima at some value of the order of $\Theta(1)$. Putting all cases together, we obtain $\textsf{SC}_c(k^*) = \Theta(n)$ and $\textsf{Var}_c(\hat \mu~|~ k^*) = \Theta(1)$ when $\alpha > \frac{1}{2}$. This concludes the proof of the theorem.     
\end{enumerate}


\section{Proofs of Results in Section~\ref{sec:autonomous}}\label{app:equilibrium}

\subsection{Proof of Lemma~\ref{lem:autonomous_eps}}
Given $(\vec c, \sigma^2)$, $S \subseteq [n]$ is some fixed coalition of size $k \geq 2$, where all players in $S$ choose their own privacy levels autonomously. Let $\vec \epsilon$ represent the vector of privacy levels chosen by the players. 

For each player $i \in S$, their burden of participation is given by: 
\[
           B_i(c_i, \vec \epsilon) = \frac{\sigma^2}{k} + \frac{2}{k^2}\left( \frac{1}{\epsilon_i^2} + \sum_{j \in S, j \neq i}\frac{1}{\epsilon_j^2} \right) + k^{\alpha} c_i \epsilon_i.
\]
Since each player $i \in S$ is burden-minimizing, their optimal choice of privacy level is obtained by solving:
\[
            \epsilon_i^* = \arg \min_{\epsilon_i} ~B_i(c_i, (\epsilon_i, \vec \epsilon_{-i})).
\]
Note that $B_i(\cdot)$ is additively separable in the $\epsilon$'s of different players, so $\epsilon_i^*$ does not depend on the decisions of other players in $S$. Further, $B_i(\cdot)$ is convex in $\epsilon_i$, so $\epsilon_i^*$ is unique and can be obtained as follows: 
\[
         \frac{\partial B_i(\cdot)}{\partial \epsilon_i} \big\vert_{\epsilon_i^*} = 0 \iff -\frac{4}{k^2 \epsilon_i^{*3} } + k^{\alpha}c_i = 0 \iff \epsilon_i^* = \left(\frac{4}{k^{(2+\alpha)}c_i } \right)^{1/3}. 
\]
Importantly, note that although $\epsilon_i^*$'s were chosen selfishly by players to minimize their individual burdens, $\vec \epsilon^*$ still constitutes the unique Nash equilibrium of the sub-game where players in $S$ try to choose their privacy levels. We can now derive the expressions for the social cost and the estimator variance at coalition $S$. Since $\vec \epsilon^*(S)$ is uniquely defined given $S$, we can define the functions $\textsf{SC}_d(S)$ and $\textsf{Var}_d(\hat \mu~|~S)$ to capture social cost and variance respectively in the decentralized setting, by dropping the dependence on the $\epsilon$-profile. Plugging in $\vec \epsilon^*(S)$, we obtain the desired expressions for $\textsf{SC}_d(S)$ and $\textsf{Var}_d(\hat \mu~|~S)$:
\begin{align*}
    \textsf{Var}_d(S) &= \frac{1}{k}\left[ \sigma^2 + \frac{2}{k}\left( \sum_{i \in S} \frac{1}{\epsilon_i^{*2} } \right) \right] \\
    &= \frac{1}{k}\left[ \sigma^2 + \frac{2}{k} \sum_{i \in S} \frac{ k^{2(2+\alpha)/3}c_i^{2/3} }{ 2^{4/3} }  \right] \\
    &= \frac{1}{k}\left[ \sigma^2 + \frac{k^{2(\alpha+2)/3} }{2^{1/3}}\left(\sum_{i \in S}c_i^{2/3}/k \right) \right].
\end{align*}
\begin{align*}
    \textsf{SC}_d(S) &= (n+1)\sigma^2 - \left[k \sigma^2 - \frac{1}{k}\sum_{i \in S}\left( \frac{2}{\epsilon_i^{*2} } + k^{1+\alpha}c_i \epsilon_i^* \right) \right] \\
    &= (n+1)\sigma^2 - \left[k\sigma^2 - \frac{1}{k}\sum_{i \in S}\left( 2\cdot \frac{ k^{2(\alpha+2)/3} c_i^{2/3} }{2^{4/3}} + k^{1+\alpha}c_i \cdot \frac{2^{2/3}}{ k^{(2+\alpha)/3}c_i^{1/3} }  \right) \right] \\
    &= (n+1)\sigma^2 - \left[k\sigma^2 - \frac{1}{k}\sum_{i \in S}\left( \frac{k^{2(\alpha+2)/3} c_i^{2/3} }{ 2^{1/3}} + 2^{2/3}\cdot k^{(2\alpha+1)/3}c_i^{2/3}  \right) \right] \\
    &= (n+1)\sigma^2 - \left[k \sigma^2 - \frac{1}{k}\cdot \frac{k^{(2\alpha+1)/3} }{2^{1/3} } \left(\sum_{i \in S} (k+2)c_i^{2/3}\right) \right] \\
    &= (n+1)\sigma^2 - \left[k \sigma^2 - \left(\frac{1}{2^{1/3}}\right)(k+2)k^{(2\alpha+1)/3}\left(\sum_{i \in S}c_i^{2/3}/k \right)   \right].
\end{align*}
This concludes the proof of the lemma.  

\subsection{Proof of Claim~\ref{clm:high_var}}
Given ($\vec c, \sigma^2$), we need to derive sufficient conditions on the number of players $n$ for the grand coalition to exist for different regimes of parameter $\alpha$. By Claim~\ref{clm:eq_conditions_Nash}, we know that the grand coalition exists (under both Nash stability and robust stability definitions) if and only if: 
\begin{align*}
       \sigma^2 \cdot \frac{(n-1)}{n^{(2\alpha+1)/3}} \geq \sum_{j \in [n]} (c_j^2/2)^{1/3} + 2(c_n^2/2)^{1/3}.
\end{align*}
When $\alpha > - \frac{1}{2}$, $n^{(2\alpha+1)/3}$ is an increasing function in $n$. Therefore, 
\begin{align*}
    &2 \leq n \leq \left( \frac{\sigma^2}{4(c_n^2/2)^{1/3}} \right)^{3/(2\alpha+1)}\\ 
    &\implies n^{(2\alpha+1)/3} \leq\frac{\sigma^2}{4(c_n^2/2)^{1/3}} \\
    &\implies n^{(2\alpha+1)/3} \leq\frac{\sigma^2}{(c_n^2/2)^{1/3}} \cdot \frac{(n-1)}{(n+2)} \quad \text{(since $\min_{n \geq 2} \frac{(n-1)}{(n+2)} = \frac{1}{4}$)} \\
    &\implies n^{(2\alpha+1)/3} \leq \frac{(n-1)\sigma^2}{ \sum_{j \in [n]} (c_j^2/2)^{1/3} + 2(c_n^2/2)^{1/3} } \quad \text{(since $\sum_{j \in [n]} (c_j^2/2)^{1/3} + 2(c_n^2/2)^{1/3} \leq (n+2)(c_n^2/2)^{1/3}$)} \\
    &\implies \sigma^2 \cdot \frac{(n-1)}{n^{(2\alpha+1)/3}} \geq \sum_{j \in [n]} (c_j^2/2)^{1/3} + 2(c_n^2/2)^{1/3}.
\end{align*}
However, when $\alpha < -\frac{1}{2}$, $n^{(2\alpha+1)/3}$ is a decreasing function of $n$. In that case, 
\begin{align*}
    &n \geq \max \left\{2, \left[ \frac{4(c_n^2/2)^{1/3}}{\sigma^2}\right]^{-3/(2\alpha+1)} \right\} \\
    &\implies n \geq \left[ \frac{4(c_n^2/2)^{1/3}}{\sigma^2}\right]^{-3/(2\alpha+1)} \\
    &\implies n^{(2\alpha+1)/3} \leq \left[ \frac{4(c_n^2/2)^{1/3}}{\sigma^2}\right]^{-1} \\
    &\implies n^{(2\alpha+1)/3} \leq \frac{\sigma^2}{4(c_n^2/2)^{1/3}} \quad \text{(remaining steps are identical)}\\
    &\implies n^{(2\alpha+1)/3} \leq\frac{\sigma^2}{(c_n^2/2)^{1/3}} \cdot \frac{(n-1)}{(n+2)}  \\
    &\implies n^{(2\alpha+1)/3} \leq \frac{(n-1)\sigma^2}{ \sum_{j \in [n]} (c_j^2/2)^{1/3} + 2(c_n^2/2)^{1/3} }\\
    &\implies \sigma^2 \cdot \frac{(n-1)}{n^{(2\alpha+1)/3}} \geq \sum_{j \in [n]} (c_j^2/2)^{1/3} + 2(c_n^2/2)^{1/3}.
\end{align*}
This concludes the proof of the claim.

\subsection{Proof of Claim~\ref{clm:equal_cost_Nash}}
Suppose, we are in the identical cost setting where in a coalition of size $k$, each player has a privacy cost $c(k) = c\cdot k^{\alpha}$. We want to show that there cannot exist any Nash equilibrium coalition of size $k$ for $2 \leq k < n$ for any value of $\sigma$. We will prove by contradiction. Suppose, there does exist one of size $k$ with $k \in \{2,...(n-1)\}$ at variance level $\sigma^2$. Then, the revised equilibrium conditions are as follows:

\paragraph{No player in $S$ wants to leave.}
\begin{align*}
      \frac{(k-1)}{k^{(2\alpha+1)/3}}  \geq  \frac{ (k+2)(c^2/2)^{1/3} }{\sigma^2}   \iff \frac{1}{k+2}\cdot \frac{(k-1)}{k^{(2\alpha+1)/3}} \geq \frac{(c^2/2)^{1/3}}{\sigma^2}. 
\end{align*}

\paragraph{No player in $S^c$ has an incentive to join.}
\begin{align*}
    &k(k+1) < \left( \frac{3(c^2/2)^{1/3}}{\sigma^2}\right) (k+1)^{(2\alpha+4)/3} + \left(\frac{k(c^2/2)^{1/3}}{\sigma^2}\right) k^{(2\alpha+4)/3} \\
    &\iff \frac{(c^2/2)^{1/3}}{\sigma^2} > \frac{k(k+1)}{ 3(k+1)^{(2\alpha+4)/3} + k^{(2\alpha+7)/3} }.
\end{align*}
The proof is completed by noting that the functions $g(x) = \frac{1}{x+2}\cdot \frac{(x-1)}{x^{(2\alpha+1)/3}}$ and $h(x) = \frac{x(x+1)}{ 3(x+1)^{(2\alpha+4)/3} + x^{(2\alpha+7)/3} }$ satisfy: 
\[
        g(x) < h(x) \quad \forall~ x \geq 2, ~\forall~\alpha \in [-1,1]. 
\]
Therefore, for any given $c$, there exists no $\sigma$ that satisfies both equilibrium conditions simultaneously, implying non-existence of Nash equilibrium coalitions of size $k\in \{2,...(n-1)\}$ for $\alpha \in [-1,1]$. This concludes the proof.

\subsection{Proof of Claim~\ref{clm:equal_cost_robust}}
Using the equilibrium conditions from Claim~\ref{clm:eq_conditions_robust}, there exists a robust equilibrium coalition of size $k$ at variance level $\sigma^2$ if and only if: 
\[
    \left(\frac{k+2}{k-1}\right)   \cdot k^{\frac{2\alpha+1}{3}} \leq 2^{1/3} \cdot \frac{\sigma^2}{c^{2/3}} <  \left(\frac{k+3}{k}\right)   \cdot (k+1)^{\frac{2\alpha+1}{3}}.
\]
Firstly, for $\alpha \leq -\frac{1}{2}$, the upper and lower bound are incompatible with each other, so clearly there exist no robust equilibrium in this regime. 

However, for $\alpha > -\frac{1}{2}$, there always exists some $k'(\alpha)$ such that for $k \geq k'(\alpha)$, the upper bound exceeds the lower bound. So, there may indeed exist some $k''$ satisfying $k'(\alpha) \leq k'' < n$ for which the above conditions hold simultaneously. This would mean that a robust equilibrium of size $k''$ would indeed exist at $\sigma^2$. However, $k''$, if it exists, would have to be unique because the permissible intervals of $\frac{\sigma^2}{c^{2/3}}$ that admit equilibria of different sizes are completely non-overlapping. This concludes the proof of the claim.

\subsection{Proof of Theorem~\ref{thm:opt_eq}}
We will complete the proof in several parts: first, we will argue about the stability of the grand coalition across regimes of $\alpha$. Then, we will try to analyze the characteristics of the social cost and estimator variance at the optimal stable coalition (if it exists) under different regimes. 

\paragraph{Stability of the grand coalition.} When $\alpha < - \frac{1}{2}$, by Claim~\ref{clm:high_var}, we know that for sufficiently large $n$, the grand coalition is always stable (under both our definitions of stability). Let us now analyze what happens when $\alpha \geq -\frac{1}{2}$. Recall that the grand coalition is stable under both definitions of stability if and only if: 
\[
             \sigma^2 \cdot \frac{(n-1)}{n^{(2\alpha+1)/3}} \geq \sum_{j \in [n]} (c_j^2/2)^{1/3} + 2(c_n^2/2)^{1/3}.   
\]
When $\alpha = -\frac{1}{2}$, the condition reduces to: 
\[
            \sigma^2 \geq \frac{1}{(n-1)} \cdot \sum_{j \in [n]} (c_j^2/2)^{1/3} + 2(c_n^2/2)^{1/3}. 
\]
Note that $(n+2)(c_{min}^2/2)^{1/3} \leq \sum_{j \in [n]} (c_j^2/2)^{1/3} + 2(c_n^2/2)^{1/3} \leq (n+2)(c_{max}^2/2)^{1/3}$. Further, for all $n \geq 2$, $1 \leq \frac{n+2}{n-1} \leq 4$. Therefore, the RHS of the above condition is always $\Theta(1)$ and whether the grand coalition exists or not depends on the exact ordering of $\sigma^2$ and the RHS. This implies that the grand coalition may or may not exist at $\alpha = -\frac{1}{2}$. However, when $\alpha > -\frac{1}{2}$, the RHS has an additional multiplicative term $n^{(2\alpha+1)/3}$ which makes it grow in $n$. Therefore, for $\alpha > -\frac{1}{2}$, for sufficiently large $n$, the grand coalition is never stable under any notion of stability.  

\paragraph{Characteristics of the Optimal Stable Coalition.} We will analyze regime-wise for $\alpha$:

Firstly, when $\alpha < -\frac{1}{2}$, we have shown that the grand coalition is always stable and has a social cost of the order $\Theta(n^{(2\alpha+4)/3})$ and estimator variance of the order $\Theta(n^{(2\alpha+1)/3})$. But the stability of the grand coalition does not rule out the existence of other coalitions which might be stable and have better social cost and estimator variance. What we do know (using the grand coalition) is that the social cost of the optimal coalition must be $\mathcal{O}(n^{(2\alpha+4)/3})$ and the variance of the optimal coalition must be $\mathcal{O}(n^{(2\alpha+1)/3})$. We will now show that these are in fact, \textit{tight}, i.e., we can construct a problem instance where the \textit{optimal stable} coalition also achieves these same orders. Consider a problem instance where all players have the identical privacy cost parameter $c$. For this special setting, we know (by Claims~\ref{clm:equal_cost_Nash} and \ref{clm:equal_cost_robust}) that there are no Nash-stable or robust-stable coalitions of intermediate sizes ($2$ through $(n-1)$). However, the grand coalition is still stable under both definitions. This implies that the grand coalition must the \textit{unique} stable coalition in this problem instance, making it also the optimal one. Thus, we have a problem instance where the social cost and the estimator variance at the optimal stable coalition are $\Theta(n^{(2\alpha+4)/3})$ and $\Theta(n^{(2\alpha+1)/3})$ respectively.   \\  

When $\alpha = -\frac{1}{2}$, suppose that at least one stable coalition exists under both notions of stability. Firstly, note that the social cost and estimator variance of the optimal stable coalition are trivially $\mathcal{O}(n)$ and $\mathcal{O}(1)$ respectively (these are achieved even by the empty coalition). We now need to show that these orders are \textit{tight} even if stable coalitions do exist. We will again construct a problem instance ($\vec c, \sigma^2$) consisting of all players with identical cost parameter $c$. But this time, we will impose additional conditions on $\sigma$ and $c$. In particular, we assume that $\sigma^2 \geq 4(c^2/2)^{1/3}$. This guarantees that the grand coalition is stable (under both Nash and robust definitions) even at $\alpha = -\frac{1}{2}$, with social cost of the order $\Theta(n)$ and variance of the order $\Theta(1)$. This must also be the unique stable coalition in this setting (again by Claims~\ref{clm:equal_cost_Nash} and \ref{clm:equal_cost_robust}) making it the optimal one. This concludes the proof of the $\alpha = -\frac{1}{2}$ case. \\

Finally, when $\alpha > - \frac{1}{2}$, the social cost and estimator variance of the optimal stable coalition are again trivially $\mathcal{O}(n)$ and $\mathcal{O}(1)$ respectively. We will again prove tightness, assuming that stable coalitions do exist in this regime. Consider the problem instance where all players have the identical cost parameter $c$. Firstly, note that the grand coalition is no longer robust-stable. In fact, by Claim~\ref{clm:equal_cost_robust}, there exists atmost one robust-stable coalition for this problem instance. If it exists, it must be of some intermediate size $k'(\alpha)$, where $k'(\alpha)$ is the unique value of $k$ which satisfies: 
\[
         \left(\frac{k+2}{k-1}\right)   \cdot k^{\frac{2\alpha+1}{3}} \leq 2^{1/3} \cdot \frac{\sigma^2}{c^{2/3}} <  \left(\frac{k+3}{k}\right)   \cdot (k+1)^{\frac{2\alpha+1}{3}}.
\]
By choosing the $(c, \sigma^2)$ appropriately, we can ensure that exactly one robust-stable coalition of size $k'(\alpha)$ does exist for our constructed problem instance. Importantly, note that $k'(\alpha)$ is independent of $n$, so when $n$ is sufficiently large, $k'(\alpha) = \Theta(1)$. Therefore, the optimal robust-stable coalition in this problem instance has social cost $\Theta(n)$ and variance $\Theta(1)$. This establishes tightness for robust-stable coalitions. 

But, what about tightness for Nash-stable coalitions? For this, we will use Claim~\ref{clm:Nash_stronger}. Since the set of all robust-stable coalitions is a superset of the set of all Nash-stable coalitions for any problem instance, the orders of social cost and variance for the optimal robust-stable coalition form lower bounds for the orders of the social cost and variance of the optimal Nash-stable coalition (provided it exists). Therefore, if a Nash-stable coalition exists, the optimal Nash-stable coalition must have social cost $\Omega(n)$  and variance $\Omega(1)$. This immediately implies that the orders are tight even for the optimal Nash-stable coalition (since social cost is also trivially $\mathcal{O}(n)$ and variance is trivially $\mathcal{O}(1)$). This concludes the proof of the theorem. 


\section{Proofs of Results in Section~\ref{sec:partial}}

\subsection{Proof of Lemma~\ref{lem:partial_general}}
We will complete the proof in parts. First, in order to determine the feasible set of $\epsilon$'s for a set $S$ to be Nash-stable, it suffices to obtain the necessary and sufficient conditions for $S$ (with $|S| = k$) to be Nash-stable at fixed privacy level $\epsilon$. Note that when $S$ represents the set of participating players, the variance of the estimator $\hat \mu$ is given by: 
\[
     \frac{\sigma^2}{k} + \frac{2}{k^2}\left(\sum_{i \in S}\frac{1}{\epsilon_i^2} \right) = \frac{1}{k}\left(\sigma^2 + \frac{2}{\epsilon^2}\right). 
\]
Therefore, when some player $i$ chooses to be part of the coalition $S$, their burden of participation is given by: 
\[
    B_i(S, \epsilon) = \frac{1}{k}\left(\sigma^2 + \frac{2}{\epsilon^2}\right) + \left(c_i k^{\alpha}\right)\cdot \epsilon.
\]

\paragraph{No player in $S$ has incentive to leave.} For all $i \in S$, we must have:
\begin{align*}
    &\frac{1}{k}\left(\sigma^2 + \frac{2}{\epsilon^2}\right) + c_ik^{\alpha} \epsilon \leq \sigma^2 \\
    \iff &\sigma^2 \cdot \left(\frac{k-1}{k}\right) \geq c_i k^{\alpha} \epsilon + \frac{2}{k\epsilon^2}\\
    \iff & \sigma^2 \geq \frac{k^{1+\alpha}}{(k-1)} \cdot c_i \epsilon + \frac{2}{(k-1)\epsilon^2}.
\end{align*}

\paragraph{No player outside $S$ has an incentive to join unilaterally.} For all $i \notin S$, we must have: 
\begin{align*}
    &\frac{1}{k+1}\left(\sigma^2 + \frac{2}{\epsilon^2}\right) + c_i(k+1)^{\alpha} \epsilon > \sigma^2 \\
    \iff & \sigma^2 < \frac{(k+1)^{1+\alpha}}{k}\cdot c_i \epsilon + \frac{2}{k\epsilon^2}.
\end{align*}
Combining both conditions, we conclude that $S \subseteq [n]$ with $n > k \geq 2$ is a Nash equilibrium coalition if and only if: 
\[
       \frac{2}{k\epsilon^2} + \frac{(k+1)^{1+\alpha}}{k}\cdot \epsilon \min_{i \notin S}c_i > \sigma^2 \geq \frac{2}{(k-1)\epsilon^2} + \frac{k^{1+\alpha}}{k-1}\cdot \epsilon \max_{i \in S}c_i. 
\]
This immediately provides us the description for $\mathcal{R}_{\epsilon}(S)$ for $k < n$. For $k = n$, we only have the one-sided condition (because there are no outside players). This concludes the first part of the proof. 

For the second part of the proof, it suffices to show that if $S$ with $|S| = k$ is Nash-stable at a fixed privacy level $\epsilon$, then $S_k$ is also Nash-stable at the same $\epsilon$. This would immediately imply that every $\epsilon$ that is feasible for $S$ must also be feasible for $S_k$, i.e., $\mathcal{R}_{\epsilon}(S) \subseteq \mathcal{R}_{\epsilon}(S_k)$. If $S = S_k$, then we are trivially done. Otherwise, if $S \neq S_k$, this means that $S$ is not downward closed and there are some missing players in $S$ with indices between $1$ and $k$ (both included). This means that: 
\[
      \min_{i \notin S}c_i \leq c_{k+1} = \min_{i \notin S_k}c_i.
\]
Similarly, 
\[
    \max_{i \in S} c_i \geq c_k = \max_{i \in S_k} c_i. 
\]
Using both inequalities, we obtain: 
\[
     \frac{2}{k\epsilon^2} + \frac{(k+1)^{1+\alpha}}{k}\cdot \epsilon \min_{i \notin S_k}c_i > \sigma^2 \geq \frac{2}{(k-1)\epsilon^2} + \frac{k^{1+\alpha}}{k-1}\cdot \epsilon \max_{i \in S_k}c_i, 
\]
which implies that $S_k$ is also Nash-stable at $\epsilon$. This concludes this part of the proof. 

The last two parts follow directly by noting that if $S$ is a Nash-stable coalition at some privacy level $\epsilon$, then the social cost and estimator variance achieved by this arrangement are given by: 
\[
    \textsf{SC}(S, \epsilon) = (n+1)\sigma^2 - \left[k\sigma^2 - \frac{2}{\epsilon^2} - k^{1+\alpha}\cdot \epsilon \cdot \left(\sum_{i \in S}c_i/k \right) \right]; 
\]
\[
     \textsf{Var}(\hat \mu~|~S, \epsilon) = \frac{1}{k}\left(\sigma^2 + \frac{2}{\epsilon^2} \right). 
\]
Out of all stable coalitions of size $k$ (if they exist), at a fixed $\epsilon$, $S_k$ has the best social cost. Further, $S_k$ also has the largest set of feasible $\epsilon$'s out of all coalitions of size $k$. Therefore, the best social cost achieved by a Nash-stable coalition of size $k$ (if it exists) must be achieved at $(S_k,\epsilon^*)$ where $\epsilon_{(k)}^* = \arg \min_{\epsilon \in \mathcal{R}_{\epsilon}(S_k)} \textsf{SC}(S_k, \epsilon)$. Finally, we plug in $\epsilon^*$ into the variance expression to retrieve the estimator variance achieved by the same coalition. This concludes the proof. Note that the notation $\textsf{SC}_f(k, \epsilon_{(k)}^*)$ (and similarly for variance) in the main result statement are defined with $k$ and $\epsilon_{(k)}^*$ as arguments to indicate that they correspond to the best achievable social cost and variance at size $k$ (and are different from the standard notation $\textsf{SC}(S, \epsilon)$). 

\subsection{Proof of Lemma~\ref{lem:partial_grand}}
We will again do the proof in two parts: first for $\alpha < \frac{1}{2}$ and then for $\alpha > \frac{1}{2}$. 

First note that $\epsilon'=\left(\frac{4}{n^{1+\alpha} \bar c_n} \right)^{1/3}$ is the value of $\epsilon$ that minimizes the social cost of the grand coalition in the absence of the stability criterion (we still do not know whether it is stable at this $\epsilon$, i.e., whether $\epsilon' \in \mathcal{R}_{\epsilon}(S_n)$). Therefore, if we can show that the grand coalition is Nash-stable at $\epsilon'$ for $\alpha < \frac{1}{2}$ and sufficiently large $n$, we immediately prove the entire first part of the lemma. So, we try to do that next. We know that the grand coalition $S_n$ is Nash-stable at any fixed privacy level $\epsilon$ if and only if: 
\[
     \sigma^2 \geq \frac{2}{(n-1)\epsilon^2} + \frac{n^{1+\alpha}}{n-1}\cdot c_n \epsilon.
\]
We first claim that if any chosen $\epsilon$ satisfies: 
\[
    \left(\frac{\sigma^2}{4c_n} \right)n^{-\alpha}   > \epsilon > \left(\frac{2\sqrt{2}}{\sigma}\right) n^{-1/2}, 
\]
then it must also satisfy $\sigma^2 \geq \frac{2}{(n-1)\epsilon^2} + \frac{n^{1+\alpha}}{n-1}\cdot \epsilon c_n$, i.e., the grand coalition must be Nash-stable at said $\epsilon$. This is easy to verify. Firstly,
\begin{align*}
    \epsilon <  \left(\frac{\sigma^2}{4c_n} \right)n^{-\alpha} &\implies \sigma^2/2 > 2 n^{\alpha} c_n \epsilon \\
    &\implies \sigma^2/2 > \frac{n^{1+\alpha}}{n/2}\cdot c_n \epsilon \\
    &\implies \sigma^2/2 > \frac{n^{1+\alpha}}{(n-1)}\cdot c_n \epsilon.
\end{align*}
Similarly, 
\begin{align*}
    \epsilon > \left(\frac{2\sqrt{2}}{\sigma}\right) n^{-1/2} &\implies \epsilon^2 > \frac{8}{n\sigma^2} \\
    &\implies \sigma^2/2 > \frac{4}{n \epsilon^2} \\
    &\implies \sigma^2/2 > \frac{2}{\epsilon^2 (n/2)} \\
    &\implies \sigma^2/2 > \frac{2}{(n-1)\epsilon^2}. 
\end{align*}
Putting both of these together, we obtain: $\sigma^2 > \frac{2}{(n-1)\epsilon^2} + \frac{n^{1+\alpha}}{n-1}\cdot \epsilon c_n$. We now need to show that for sufficiently large $n$ and $\alpha < \frac{1}{2}$, $\epsilon'$ satisfies:
\[
     \left(\frac{\sigma^2}{4c_n} \right)n^{-\alpha} > \epsilon' > \left(\frac{2\sqrt{2}}{\sigma}\right) n^{-1/2}. 
\]
Recall that $\epsilon' = \left(4/\bar c_n \right)^{1/3} n^{-(1+\alpha)/3}$. First, note that for $\alpha < \frac{1}{2}$, we must have:
\[
       n^{-\alpha} > n^{-(1+\alpha)/3} > n^{-1/2}~\forall~n \in \mathbb{N}, n \geq 2.
\]
Even with arbitrary positive constant multipliers $C_1, C_2, C_3$, for sufficiently large $n$, we will still have:
\[
      C_1\cdot n^{-\alpha} > C_2 \cdot n^{-(1+\alpha)/3} > C_3 \cdot n^{-1/2}.
\]
Substituting $C_1 = \left(\frac{\sigma^2}{4c_n} \right)$, $C_2 = \left( \frac{4}{\bar c_n}\right)^{1/3}$ and $C_3 = \left(\frac{2\sqrt{2}}{\sigma} \right)$, we conclude that
$S_n$ is Nash-stable at $\epsilon'$ for sufficiently large $n$ and $\alpha < \frac{1}{2}$. We have already argued that $\epsilon' \in \mathcal{R}_{\epsilon}(S_n) \implies \epsilon^* = \epsilon'$. 

To prove the second part of the lemma, we will use the necessary and sufficient conditions for Nash-stability in the partially decentralized setting. First, note that the grand coalition can never be Nash-stable in the $\alpha > \frac{1}{2}$ regime. This is because in order for the grand coalition to be stable, we must have:
\[
     \sigma^2 \geq \frac{2}{(n-1)\epsilon^2} + \frac{n^{1+\alpha}}{n-1}\cdot \epsilon c_n.
\]
Note that $\sigma^2 = \Theta(1)$, but for any choice of $\epsilon$, the RHS grows polynomially in $n$, hence the contradiction. However, we need to check if we can still have stable coalitions of some other size $k < n$. In that case, there must exist a feasible choice of $\epsilon$ for which:
\[
    \frac{2}{k\epsilon^2} + \frac{(k+1)^{1+\alpha}}{k}\cdot \epsilon c_{k+1} > \sigma^2 \geq \frac{2}{(k-1)\epsilon^2} + \frac{k^{1+\alpha}}{k-1}\cdot \epsilon c_k. 
\]
Just working with second inequality, we obtain:
\[
      \sigma^2 > \frac{k^{1+\alpha}}{k-1}\cdot \epsilon c_k \implies \sigma^2 > k^{\alpha} \cdot \epsilon \cdot c_{min} \implies \epsilon < \left(\frac{\sigma^2}{c_{min}} \right)k^{-\alpha}.
\]
Simultaneously, we have:
\[
      \sigma^2 > \frac{2}{(k-1)\epsilon^2} > \frac{2}{k\epsilon^2} \implies \epsilon > k^{-1/2}\sqrt{\frac{2}{\sigma^2} }. 
\]
Putting both inequalities together, we obtain that in order for a coalition of size $k$ to be stable, $\epsilon$ must satisfy:
\[
      C_3 \cdot k^{-1/2} < \epsilon < C_4 \cdot k^{-\alpha}.
\]
However, we are in the regime $\alpha > \frac{1}{2}$ (so the lower bound has a higher order than the upper bound). This implies that the above condition can only be satisfied up to some threshold $k = g(\alpha, C_3, C_4)$. Therefore, any stable coalition (if it exists) must be of size $\Theta(1)$ and similarly, any feasible $\epsilon$ which induces such a stable coalition must also be $\Theta(1)$. This concludes the proof.

\subsection{Proof of Theorem~\ref{thm:partial_cent}}
We already have the major parts of the proof in place using Lemma~\ref{lem:partial_grand}. We will reason separately for the regimes $\alpha < \frac{1}{2}$, $\alpha = \frac{1}{2}$ and $\alpha > \frac{1}{2}$. 

\paragraph{$\alpha > \frac{1}{2}$ regime.} The result in this regime follows directly from Lemma~\ref{lem:partial_grand}. We showed there that any Nash-stable coalition that exists must be of size $\Theta(1)$ and must exist at $\epsilon$ of size $\Theta(1)$. This implies that social cost of the optimal stable coalition (if it exists) must be of the order $\Theta(n)$ and similarly, its estimator variance must be $\Theta(1)$. 

\paragraph{$\alpha = \frac{1}{2}$ regime.} For $\alpha = \frac{1}{2}$, we will argue case by case. Firstly, re-using some of the proof steps from the first part of Lemma~\ref{lem:partial_grand}, in order for the grand coalition to be stable, we need to pick $\epsilon$ such that: 
\[
     C_1 \cdot n^{-1/2} < \epsilon < C_2 \cdot n^{-\alpha}, 
\]
where $C_1, C_2$ are some positive constants. At $\alpha = \frac{1}{2}$, any solution to the above system (if it exists) must satisfy $\epsilon = \Theta(n^{-1/2})$. This implies that if the grand coalition is indeed Nash-stable, $\epsilon$ is always of the order $\Theta(n^{-1/2})$ which achieves social cost of the order $\Theta(n)$ and estimator variance of the order $\Theta(1)$. Using an identical argument as the second part of Lemma~\ref{lem:partial_grand}, if there exists any other stable coalition of size $k < n$, all feasible $\epsilon$ which induce these stable coalitions must also be of order $\Theta(k^{-1/2})$. As such, all such stable coalitions will achieve social cost of order $\Theta(n)$ and estimator variance of order $\Theta(1)$. Putting all cases together, the optimal stable coalition in the $\alpha = \frac{1}{2}$ regime (if any exists) must achieve social cost of the order $\Theta(n)$ and estimator variance of the order $\Theta(1)$.\\

\paragraph{$\alpha < \frac{1}{2}$ regime.} We have already argued that for sufficiently large $n$ and for $\alpha < \frac{1}{2}$, the grand coalition is Nash-stable at $\epsilon = \left(\frac{4}{n^{1+\alpha} \bar c_n} \right)^{1/3} = \Theta(n^{-(1+\alpha)/3})$. If the grand coalition is the \textit{optimal} Nash-stable coalition, we are done. Now, suppose that the grand coalition is not the optimal Nash-stable coalition. We first argue that the optimal stable coalition must still be of size at least $n-o(n)$. This follows from the fact that the social cost achieved by the optimal Nash-stable coalition is known to be $\Theta(n^{2(1+\alpha)/3})$ (this is because the grand coalition already achieves this order, which matches the order of social cost under the fully centralized mechanism --- the best across mechanisms. So the optimal stable coalition must also achieve the same order). We prove by contradiction. If the size of the optimal stable coalition is any smaller than $n-o(n)$, i.e., if the size is $\Theta(n)$ or $o(n)$, then the social cost immediately becomes order $\Theta(n)$ (which is worse than the promised sublinear order $\Theta(n^{2(1+\alpha)/3})$), irrespective of the choice of $\epsilon$. This is easy to see: 
\[
    k = o(n) \implies \textsf{SC}_f(k, \epsilon) = (n+1-k)\sigma^2 + \frac{2}{\epsilon^2} + k^{1+\alpha}\cdot \epsilon \cdot \Theta(1) = \Theta(n). 
\]
\[
    k = \zeta n ~\text{for }\zeta \in (0, 1) \implies \textsf{SC}_f(k, \epsilon) = (1-\zeta)n \sigma^2 +\sigma^2 + \frac{2}{\epsilon^2} + k^{1+\alpha}\cdot \epsilon \cdot \Theta(1) = \Theta(n).
\]
This shows that if the optimal stable coalition is not of size $n$, it must be of size at least $n-o(n)$. Plugging in the size, the corresponding social cost would be given by: 
\[
     o(n) \sigma^2 + \frac{2}{\epsilon^2} + \Theta(n^{1+\alpha})\cdot \epsilon,
\]
which we know must be of order $\Theta(n^{2(1+\alpha)/3})$. This means that all terms in the expression must be $\mathcal{O}(n^{2(1+\alpha)/3})$ which further implies: 
\[
    \Theta(n^{1+\alpha})\cdot \epsilon = \mathcal{O}(n^{2(1+\alpha)/3}) \implies \epsilon = \mathcal{O}(n^{-(1+\alpha)/3}), \quad \text{and}
\]
\[
     \frac{1}{\epsilon^2} = \mathcal{O}(n^{2(1+\alpha)/3}) \implies \epsilon = \Omega(n^{-(1+\alpha)/3}).
\]
Together, this implies that at the optimal Nash-stable coalition, $\epsilon = \Theta(n^{-(1+\alpha)/3})$. Note that this again is the ``desired" order for $\epsilon$, so it must achieve the desired orders for the social cost and estimator variance. This concludes the proof of the theorem.


\section{Proofs of Results in Section~\ref{sec:pos}}
\subsection{Proof of Theorem~\ref{thm:pos}}
The orders of the price of stability with respect to social cost in the $\alpha \in \left[-1, -\frac{1}{2}\right]$ regime and the $\alpha \in \left(-\frac{1}{2}, \frac{1}{2}\right]$ regimes follow directly from Theorems~\ref{thm:kstar_c} and \ref{thm:opt_eq}. Note that although Theorem~\ref{thm:opt_eq} needs to assume existence of stable coalitions to state results about the order of the social cost for the optimal coalition, we don't need any such assumptions for the price of stability analysis. This is because by definition of $\textsf{PoS}$, if stable coalitions do not exist, we use the social cost of the empty coalition case which has order $\Theta(n)$.

For the last regime ($\alpha \in \left(\frac{1}{2}, 1 \right]$), the two theorems tell us that the order of the price of stability will be $\Theta(1)$. Further, we know that $\textsf{PoS}(\textsf{SC}) \geq 1$. Therefore, in order to complete the proof, it suffices for us to show that in this $\alpha$ regime, $\textsf{PoS}(\textsf{SC})$ is also upper bounded by a constant independent of $n$. 

Firstly, note that the social cost in the decentralized setting is trivially upper bounded by $n\sigma^2$ for $\alpha \in \left(\frac{1}{2}, 1\right]$ (even if no stable equilibrium exists). We will now construct a lower bound on the social cost of the optimal coalition in the decentralized setting for the same $\alpha$ regime. Recall that the social cost at the optimal coalition of size $k\geq 2$ in the fully centralized setting for a general cost profile $\vec c$ and variance $\sigma^2$ is given by:
\[
     \textsf{SC}_c(k) = n\sigma^2 + \underbrace{\left[ \left(\frac{3}{2^{1/3}} \right) k^{2(1+\alpha)/3}\left(\frac{\sum_{i \in S_k}c_i^{2/3}}{k}\right) - (k-1)\sigma^2   \right]}_{g(k)}
\]
Therefore, the best social cost achievable overall is given by:
\[
        \textsf{SC}_c(k^*) = \min \left[n \sigma^2, \min_{k \in \{2,3..n\} } n\sigma^2 + g(k) \right],
\]
where $k^*$ is defined as per Equation~\eqref{opt:kstar}. 
Now, for all $k \in \mathbb{N}$, $2 \leq k \leq n$,  we have:
\[
      g(k) \geq h(k) = \left[ \left(\frac{3}{2^{1/3}} \right) k^{2(1+\alpha)/3}c_{min}^{2/3} - (k-1)\sigma^2  \right],
\]
with $c_{min} > 0$.This implies that: 
\[
       \min_{k \in \{2,3..n\}} g(k) \geq  \min_{k \in [2,n]} h(k). 
\]
Now, $h(x)$ is a convex function in $x$ over $\mathbb{R}$ (because $\alpha > \frac{1}{2}$) which implies that $h(x)$ has a global minimizer $x'$ satisfying:
\[
     \left(\frac{3}{2^{1/3}}\right)\cdot \frac{2(1+\alpha)}{3}\cdot x'^{(2\alpha-1)/3} \cdot c_{min}^{2/3} = \sigma^2 \iff x' = \left( \frac{\sigma^2 }{(1+\alpha)(2c_{min})^{2/3}} \right)^{ 3/(2\alpha-1) }.
\]
Based on the location of $x'$, we can now have $3$ different sub-cases:
\begin{enumerate}
    \item $x' \leq 2$: In this setting, we have:
    \begin{align*}
    \min_{k \geq 2}~h(k) = h(2) &= 3c_{min}^{2/3}2^{(2\alpha+1)/3} - \sigma^2\\ 
    &\geq 3c_{min}^{2/3}2^{(2\alpha+1)/3} - 2^{2/3}\cdot 2^{(2\alpha-1)/3}(1+\alpha)c_{min}^{2/3} \quad (\text{using $x' \leq 2$})\\
    &= 2^{(2\alpha+1)/3}(3-(1+\alpha))c_{min}^{2/3} > 0. 
    \end{align*}
    This implies that $\textsf{SC}_c(k^*) = n\sigma^2$ and therefore, $\textsf{PoS}(\textsf{SC}) = 1$.
    \item $2 < x'\leq n$: In this setting, we have $\min_{k \geq 2}~h(k) = h(x')$. Therefore, 
    \begin{align*}
        n\sigma^2 + \min_{k \in \{2,3...n\}}~g(k) &\geq (n+1)\sigma^2 + \frac{3\sigma^2}{2(1+\alpha)}\cdot x' - \sigma^2 x' \\
        &= (n+1)\sigma^2 - x' \sigma^2 \left[1 - \frac{3}{2(1+\alpha)} \right]\\
        &> (n+1)\sigma^2 - n\sigma^2 \left[1 - \frac{3}{2(1+\alpha)} \right]\\
        &= \sigma^2 \left[1+\frac{3n}{2(1+\alpha)} \right] \\
        &> \frac{3n\sigma^2}{2(1+\alpha)}.  \quad \text{(note, this lower bound $< n\sigma^2$)} 
    \end{align*}
    Therefore, $\textsf{SC}_c(k^*) > \frac{3n\sigma^2}{2(1+\alpha)}$ which implies $\textsf{PoS}(\textsf{SC}) < \frac{n\sigma^2}{ (3n\sigma^2)/(2+2\alpha)  } = \frac{2(1+\alpha)}{3} < \frac{4}{3}$. This is the upper bound that will hold when $n$ is sufficiently large. 
    \item $x' > n$: Finally, in this setting, we have: 
    \begin{align*}
        \textsf{SC}_c(k^*) &> n\sigma^2 + h(n) \\
        &= \sigma^2 + \left(\frac{3}{2^{1/3}}\right)n^{2(1+\alpha)/3}c_{min}^{2/3}. 
    \end{align*}
    This implies that: 
    \begin{align*}
        \textsf{PoS}(\textsf{SC}) &< \frac{n \sigma^2}{ \sigma^2 + \left(\frac{3}{2^{1/3}}\right)n^{2(1+\alpha)/3}c_{min}^{2/3} } \\
        &< \frac{n \sigma^2}{ \left(\frac{3}{2^{1/3}}\right)n^{2(1+\alpha)/3}c_{min}^{2/3} } \\
        &= \left(\frac{2^{1/3}\sigma^2}{3c_{min}^{2/3}} \right)\cdot n^{(1-2\alpha)/3} \quad \text{(this expression $>1$ when $x' > n$)}\\
        &< \left(\frac{2^{1/3}\sigma^2}{3c_{min}^{2/3}} \right)\cdot 2^{(1-2\alpha)/3}  \quad (\text{since $\alpha > 1/2 \implies (2\alpha-1)< 0$}) \\
        &< \left(\frac{2^{1/3}\sigma^2}{3c_{min}^{2/3}} \right). 
    \end{align*}
\end{enumerate}
Putting all three sub-cases together, when $\alpha \in \left(\frac{1}{2}, 1\right]$, the price of stability with respect to social cost satisfies: 
\[
       \textsf{PoS}(\textsf{SC}) < \max \left[\frac{4}{3}, \left(\frac{2^{1/3}\sigma^2}{3c_{min}^{2/3}} \right) \right].
\]

\subsection{Proof of Theorem~\ref{thm:pos_var}}
The proof of Theorem~\ref{thm:pos_var} follows directly from Theorems~\ref{thm:kstar_c} and \ref{thm:opt_eq} --- we use the orders of the estimator variance obtained for the optimal coalition in the centralized setting (Theorem~\ref{thm:kstar_c}) and the optimal \textit{stable} coalition in the decentralized setting (Theorem~\ref{thm:opt_eq}) to complete the proof. Again, we note that although Theorem~\ref{thm:opt_eq} needs to assume existence of stable coalitions to state results about the order of the estimator variance for the optimal coalition, we don't need any such assumptions for the price of stability analysis. This is because by definition of $\textsf{PoS}$, if stable coalitions do not exist, we use the variance of the empty coalition case which has order $\Theta(1)$.

\subsection{Deriving the constant for \textsf{PoS(SC)} in  Theorem~\ref{thm:pos_partial} for $\alpha < \frac{1}{2}$ and large $n$}
For $\alpha < \frac{1}{2}$ and large $n$, we can compute an upper bound on $\textsf{PoS(SC)}$ as follows. We know that in this regime, the social cost of the optimal stable coalition under partial decentralization must be at least as good as the social cost achieved by the grand coalition (which we showed to be Nash-stable for appropriately chosen $\epsilon$). Thus, the numerator is upper bounded by: 
\[
     \textsf{SC}_f(n, \epsilon_{(n)}^*) = \sigma^2 + \left(\frac{3}{2^{1/3}}\right) n^{2(1+\alpha)/3} \bar c_n^{2/3}. 
\]
At the same time, the denominator exactly equals the social cost achieved by the grand coalition for large $n$ (because for large $n$, the grand coalition is optimal under full centralization). The denominator evaluates to: 
\[
    \textsf{SC}_c(n) = \sigma^2 + \left(\frac{3}{2^{1/3}}\right) n^{2(1+\alpha)/3} \left(\sum_{i\in S_n} c_i^{2/3}/n \right). 
\]
Quick sanity check: the numerator exceeds the denominator. We can verify this using Jensen's inequality for the concave function $f(x) = x^{2/3}$: 
\begin{align*}
    \sum_{i \in S_n}c_i^{2/3}/n = \sum_{i=1}^n \frac{1}{n}f(c_i) \leq f\left( \sum_{i=1}^n c_i/n\right) = f(\bar c_n) = (\bar c_n)^{2/3}. 
\end{align*}
Finally, putting everything together, we have: 
\begin{align*}
    \textsf{PoS(SC)} &\leq \frac{ \sigma^2 + \left(\frac{3}{2^{1/3}}\right) n^{2(1+\alpha)/3} \bar c_n^{2/3} }{ \sigma^2 + \left(\frac{3}{2^{1/3}}\right) n^{2(1+\alpha)/3} \left(\sum_{i\in S_n} c_i^{2/3}/n \right) } \\
    &< \frac{ \left(\frac{3}{2^{1/3}}\right) n^{2(1+\alpha)/3} \bar c_n^{2/3} }{ \left(\frac{3}{2^{1/3}}\right) n^{2(1+\alpha)/3} \left(\sum_{i\in S_n} c_i^{2/3}/n \right) } \\
    &= \frac{ \bar c_n^{2/3} }{ \left(\sum_{i\in S_n} c_i^{2/3}/n \right) } \leq \frac{ c_{max}^{2/3} }{ c_{min}^{2/3} }. 
\end{align*}
Note that the last step follows only because $c_{min}$ is bounded away from $0$. (This does not work in the case where exactly one $c_i$ takes a positive value and all other $c_i$'s are zero, in which case the last step has an upper bound that would grow as $n^{1/3}$. However, we do not run into such issues here). 

\end{document}